\documentclass[11pt]{article}

\usepackage[T1]{fontenc}
\usepackage{lmodern}
\usepackage{microtype}
\usepackage[margin=1in]{geometry}
\usepackage{amsmath,amssymb,amsthm,mathtools}
\usepackage{aliascnt}
\usepackage{bm}
\usepackage{enumitem}
\usepackage{tikz}
\usetikzlibrary{arrows.meta,positioning,calc}
\usepackage{xcolor}
\usepackage[colorlinks=true,linkcolor=blue!55!black,citecolor=blue!55!black,urlcolor=blue!55!black]{hyperref}
\usepackage[nameinlink,capitalize,noabbrev]{cleveref}
\usepackage{tcolorbox}
\usepackage{bm}
\usepackage{mathrsfs}

\allowdisplaybreaks
\setlist[itemize]{leftmargin=1.5em,itemsep=0.25em,topsep=0.35em}
\setlist[enumerate]{leftmargin=1.7em,itemsep=0.25em,topsep=0.35em}

\newtheorem{theorem}{Theorem}[section]
  \newtheorem{result}{Result}
\newaliascnt{lemma}{theorem}
\newtheorem{lemma}[lemma]{Lemma}
\aliascntresetthe{lemma}

\newaliascnt{proposition}{theorem}

\aliascntresetthe{proposition}

\newaliascnt{corollary}{theorem}
\newtheorem{corollary}[corollary]{Corollary}
\aliascntresetthe{corollary}

\newaliascnt{claim}{theorem}

\aliascntresetthe{claim}

\theoremstyle{definition}
\newaliascnt{definition}{theorem}
\newtheorem{definition}[definition]{Definition}
\aliascntresetthe{definition}

\theoremstyle{remark}
\newaliascnt{remark}{theorem}
\newtheorem{remark}[remark]{Remark}
\aliascntresetthe{remark}

\newcommand{\eps}{\varepsilon}
\newcommand{\op}{\mathrm{op}}
\newcommand{\var}{\mathrm{Var}}

\newcommand{\N}{\mathbb N}
\newcommand{\cH}{\mathcal H}
\newcommand{\cA}{\mathcal A}

\newcommand{\ket}[1]{\lvert #1\rangle}
\newcommand{\bra}[1]{\langle #1\rvert}
\newcommand{\braket}[2]{\langle #1\mid #2\rangle}
\newcommand{\proj}[1]{\lvert #1\rangle\!\langle #1\rvert}
\newcommand{\Tr}{\operatorname{Tr}}
\newcommand{\dist}{\operatorname{dist}}
\newcommand{\supp}{\operatorname{supp}}
\newcommand{\id}{\mathrm{i}}
\newcommand{\E}{\mathbb E}

\newcommand{\R}{\mathbb R}
\newcommand{\C}{\mathbb C}
\newcommand{\norm}[1]{\left\lVert #1\right\rVert}
\newcommand{\abs}[1]{\left\lvert #1\right\rvert}
\newcommand{\ip}[2]{\left\langle #1,#2\right\rangle}
\newcommand{\dd}{\,\mathrm{d}}
\newcommand{\Pauli}{\mathcal P_n}

\title{Hamiltonian locality testing and certification\\ do not achieve the Heisenberg limit}

\author{Francisco Escudero Gutiérrez\thanks{Email: \href{mailto:francisco.escudero-gutierrez@inria.fr}{francisco.escudero-gutierrez@inria.fr}}\\ \normalsize{Inria-Saclay}  \and 
Junseo Lee\thanks{Email: \href{mailto:junseolee@fas.harvard.edu}{junseolee@fas.harvard.edu}} \\ \normalsize{Harvard University}  \and 
Sebastian Zur\thanks{Email: \href{mailto:zursebastian@gmail.com}{zursebastian@gmail.com}} \\ \normalsize{CNRS, Université Paris Cité, IRIF}}

\date{}

\begin{document}
\maketitle
\begin{abstract}
   We establish lower bounds for Hamiltonian property testing with access to the time-evolution operator but not its inverse. Each experiment may query the time-evolution operator multiple times, and distances between Hamiltonians are measured in the normalized Frobenius norm.
   
    In this model, we show that testing whether a Hamiltonian is $k$-local or $\varepsilon$-far from every $k$-local Hamiltonian requires $\Omega(1/\varepsilon^2)$ total evolution time, matching the upper bound of Kallaugher and Liang (TQC’25). We also prove that testing whether an unknown Hamiltonian equals a target   Hamiltonian or is $\varepsilon$-far from it requires $\Omega(1/\varepsilon^2)$ total evolution time, matching the upper bound of Sinha and Tong (2025). These are the first lower bounds for natural problems in Hamiltonian learning and testing that rule out Heisenberg-limited scaling of $1/\varepsilon$.
    
    As a third result, we show that amplitude estimation to precision $\varepsilon$ requires $\Omega(1/\varepsilon^2)$ total time evolution, recovering the result of Tang and Wright (QIP’26) in the continuous-time query model. All three results follow from the hardness of distinguishing the zero Hamiltonian from a suitably chosen ensemble of random Hamiltonians. We establish this hardness by adapting the continuous-time adversary method to forward Hamiltonian evolution.
\end{abstract}

\tableofcontents

\section{Introduction}
In recent years, understanding the Hamiltonians that model the interactions between particles in emerging quantum devices has become an active research topic. In particular, a rapidly growing literature studies \emph{Hamiltonian learning}, which aims to characterize an unknown Hamiltonian~\cite{Silva2011Practical,holzapfel2015scalable,Zubida2021Optimal,haah2022optimal,yu2023robust,Dutkiewicz.2023,huang2023heisenberg,li2023heisenberglimited,franca2024efficient,Gu2022Practical,zhao2025learning,hu2025ansatz,anshu2021sample,onorati2023efficient,rouze2023learning,bakshi2023learning,ma2024learning,arunachalam2025testing,caro2023learning,sinha2025improvedhamiltonianlearningsparsity,chen2025learning,abbas2025nearly,shin2026heisenberg, shin2026emergent, shin2026hamiltonian, he2026optimal, chen2025lower, lewis2026learning, zhou2026optimal,zhou2026characterizing,romanov2026learning,arad2026near,mobus2026robust}, and \emph{Hamiltonian testing}, which asks whether the Hamiltonian satisfies a given property~\cite{aharonov2022quantum,laborde2022quantum,she2022unitary,bluhm2024hamiltonian,arunachalam2025testing,kallaugher2025hamiltonian,gao2025quantum,sinha2025improvedhamiltonianlearningsparsity,bluhm2026certifying,cai2026optimal,flammia2026autonomous}. Given that the Hamiltonian determines the time evolution of the system, many of those works assume access to the Hamiltonian~$H$ via experiments like those in \cref{fig:access-model}. In those experiments, the algorithm makes several queries to the time evolution operator~$U_H(t),$ only for positive times $t>0$. One of the main focuses of this research line has been to devise protocols with total evolution time scaling like~$1/\eps$ with the precision parameter~$\eps$. Such a scaling is known as the \emph{Heisenberg limit}, and determining in which quantum scenarios it can be attained is also a topic of interest in query complexity~\cite{tang2025amplitude}. So far there are plenty of Hamiltonian learning and testing protocols that achieve the Heisenberg limit, but it is not known whether there is a problem that cannot attain the Heisenberg limit, with the access model of \cref{fig:access-model}. Thus, a natural question is 

\begin{quote}
\centering \emph{Can the Heisenberg limit always be attained in Hamiltonian learning and testing?}
\end{quote}\vspace{0.15cm}

\noindent In this work, we answer this question negatively by showing that Hamiltonian certification and locality testing do not attain the Heisenberg limit. As a byproduct, we also obtain an amplitude-estimation lower bound that holds with continuous-time forward access and strengthens the dimension dependence of the bound of Tang and Wright~\cite{tang2025amplitude}. All three results follow from the hardness of distinguishing the zero Hamiltonian from a suitably chosen ensemble of random Hamiltonians. We establish this hardness by adapting the continuous-time adversary method to forward Hamiltonian evolution, which may be of independent interest and may lead to new lower bounds.

\begin{figure}[ht!]
    \centering
    \includegraphics[width=0.9\linewidth]{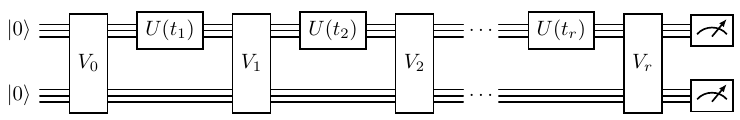}
    \caption{Experiment to access the Hamiltonian.}
    \label{fig:access-model}
\end{figure}
\subsection{Our results}
In this work, we consider $n$-qubit Hamiltonians $H$ and assume that they have bounded operator norm $\norm{H}_{\op}\leq 1$. Such a normalization is common in the literature, and allows us to focus on the scaling of the lower bounds with the precision parameter $\eps$. To define locality testing and certification, we endow the space of Hamiltonians with the distance induced by the normalized Frobenius norm, i.e. 
\begin{equation*}
    d(H,H')=\norm{H-H'}_2=\sqrt{\frac{1}{2^n}\Tr[(H-H')^2]}.
\end{equation*}
This distance can be interpreted as an average-case measure of closeness between Hamiltonians, and it was considered in prior Hamiltonian testing works~\cite{bluhm2024hamiltonian,arunachalam2025testing,kallaugher2025hamiltonian,gao2025quantum,sinha2025improvedhamiltonianlearningsparsity,bluhm2026certifying,flammia2026autonomous}. 

\subsubsection{Hamiltonian locality testing does not achieve the Heisenberg limit}
Our first result is the hardness of locality testing. The problem of locality testing is to decide whether a Hamiltonian is $k$-local or $\eps$-far from every $k$-local Hamiltonian. In this regard, we show the following result, whose formal statement can be found in \cref{cor:hardnesslocality}. 

\begin{result}\label{res:hardnesslocality}
    Let $0 < \eps < 1/4,$ $n\in \N$ and $k\in \N$. Assume that $k\leq n/8$ and $n \geq 7$. Then, an algorithm that tests whether a Hamiltonian is $k$-local or $\eps$-far from every $k$-local Hamiltonian, with constant success probability, requires $\Omega(\min\{1/\eps^2,2^{n/2}/\eps\})$ evolution time. 
\end{result}

\cref{res:hardnesslocality} is optimal, due to a matching upper bound by Kallaugher and Liang~\cite{kallaugher2025hamiltonian}. Thus, it culminates a series of works that either asked about the complexity of Hamiltonian locality testing \cite{montanaro2016survey,she2022unitary} or proposed upper bounds for the problem~\cite{bluhm2024hamiltonian,arunachalam2025testing,kallaugher2025hamiltonian}. A restriction relating of the kind $(1/\eps)\leq 2^{cn}$, for some constant $c,$ is unavoidable to prove a lower bound of~$\Omega(1/\eps^2)$ evolution time in~\cref{res:hardnesslocality}. One can always learn the Hamiltonian in normalized Frobenius norm up to error $\eps$ with $O(8^n/\eps)$ evolution time. This follows from the fact that one can learn a single Pauli coefficient with $O(1/\eps)$ evolution time (see for instance \cite{odake2025universal}).

\subsubsection{Hamiltonian certification does not achieve the Heisenberg limit}
Our second result is the hardness of Hamiltonian certification. Given a target Hamiltonian $H_0$, the problem of certification is to decide whether an unknown Hamiltonian $H$ equals $H_0$ or is $\eps$-far from $H_0$. In this regard, we show the following result, whose formal statement can be found in \cref{cor:hardnesscertification}. 

\begin{result}\label{res:hardnesscertification}
    Let $0 < \eps < 1/8$ and $n\in \N$. Assume that $n \geq 7$. Then, an algorithm that performs Hamiltonian certification with distance threshold $\eps$ and with constant success probability requires $\Omega(\min\{1/\eps^2,2^{n/2}/\eps\})$ evolution time. 
\end{result}

\cref{res:hardnesscertification} is optimal, due to a matching upper bound by Sinha and Tong~\cite{sinha2025improvedhamiltonianlearningsparsity}. Prior work showed that Heisenberg limited scaling is possible if both the target Hamiltonian $H_0$ and $H$ are assumed to be either local or sparse in the Pauli basis \cite{gao2025quantum,bluhm2026certifying}. In particular, \cite{bluhm2026certifying} proved that for Hamiltonians with constant locality, Hamiltonian certification can be done with $O(1/\eps)$ evolution time, and the authors asked whether that was possible for unstructured Hamiltonians. \cref{res:hardnesscertification} answers that question in the negative, and shows that the locality assumption is not a mere artifact of their proof. For the same reason as in the case of locality testing, an assumption of the kind $(1/\eps)=2^{O(n)}$ is also unavoidable to prove a lower bound of $\Omega(1/\eps^2)$ evolution time in~\cref{res:hardnesscertification}.

\subsubsection{Quantum amplitude estimation requires inverses, even with fractional~queries}
Our third result improves the amplitude-estimation lower bound of Tang and Wright. They showed that amplitude estimation to additive error $\eps$, using only forward queries to an $n$-qubit unitary $U$, requires $\Omega(\min\{2^n,1/\eps^2\})$ queries in the worst case. We prove the lower bound $\Omega(\min\{1/\eps^2,2^{n/2}/\eps\})$, even with continuous-time forward access to the Hamiltonian generating the unitary. In particular, it strengthens their dimension dependence in the high-precision regime $\eps<2^{-n/2}$.

More precisely, let $U := U_H(1)=e^{-\id H}$, where $\norm{H}_{\op}\leq 1$. Then each query to $U$ corresponds to one unit of evolution under $H$. We allow the algorithm to choose arbitrary forward evolution times, querying $e^{-\id Ht}$ for any $t\geq0$, and charge each query its duration $t$. More generally, our lower
bound holds in the continuous-time model of \cref{sec:conttimemodel}. Thus, total evolution time replaces the number of unitary queries as the measure of cost.

\begin{result}\label{res:hardnessofAE}
    Let $0 < \eps < 1/16$ and $n \in \N$. Assume that $n \geq 7$. Amplitude estimation to additive error $\eps$ with constant success probability requires in the worst case $\Omega(\min\{1/\eps^2,2^{n/2}/\eps\})$ total evolution time. 
\end{result}

\subsubsection{Our main technical result}
\cref{res:hardnesscertification,res:hardnesslocality,res:hardnessofAE} follow from the hardness of distinguishing the zero Hamiltonian from a random ensemble of Hamiltonians with suitable moments. More precisely, we construct a family of random diagonal Hamiltonians 
\begin{equation*}
    H=\sum_{x\in\{0,1\}^n}h_x\ket{x}\bra x
\end{equation*}
where $h_x$ are iid random variables satisfying 
\begin{equation}\label{eq:moments}
    |h_x|\leq 1,\quad \E [h_x]=0\quad \mathrm{and}\quad \E[h_x^2]=\eps^2.
\end{equation}
It is easy to see that such a family of Hamiltonians is far from the zero Hamiltonian in expectation, 
\begin{equation*}
    \E[\norm{H}_2^2]=\frac{1}{2^n}\E\left[\sum_{x\in\{0,1\}^n}h_x^2\right]=\eps^2.
\end{equation*}
From there, using concentration inequalities one can show that such a family is far from the zero Hamiltonian with constant probability, as long as $2^n\eps^2$ is a sufficiently large constant. Thus, distinguishing that random ensemble from the zero Hamiltonian is no harder than Hamiltonian certification. A similar reduction is also possible for locality testing and amplitude estimation. 

We note that taking $h_x$ to be drawn from the uniform distribution in $\{-\eps,\eps\}$ suffices to satisfy \cref{eq:moments}. However, the resulting Hamiltonian ensemble is easy to distinguish from the zero Hamiltonian. That is because for any of the Hamiltonians in the ensemble, $U_H(\pi/\eps)=-I$, while $U_0(\pi/\eps)=I$. By simulating the controlled access to $U_H(t),$ which can be done with no time evolution overhead, it is easy to distinguish these two cases with a Hadamard test (see \cref{sec:hardinst} for more details). Hence, the difficulty here is to prove that there exists a random ensemble that satisfies \cref{eq:moments} which is also difficult to distinguish from the zero Hamiltonian. Such an existence result is our main technical contribution, and we prove it via the adversary method of quantum query complexity.

\begin{result}\label{res:maintechnicalres}
    For every $0<\eps<1/2$, there exists a distribution satisfying \cref{eq:moments} such that distinguishing the induced Hamiltonian ensemble from the zero Hamiltonian requires $\Omega(1/\eps^2)$ total evolution time.
\end{result}

\subsection{Technical overview}
In this subsection, we give an idea of the proof of our main technical result, \cref{res:maintechnicalres}. First of all, we notice that when one considers Hamiltonians $F=\sum_xf(x)\ket{x}\bra{x}$ with $f:\{0,1\}^n\to\{0,1\}$, then the access model to $F$ via experiments like in \cref{fig:access-model} exactly corresponds to fractional queries to the phase oracle for $f$: the phase $e^{-\id t f(x)}$ is adjusted through the evolution time $t$. Although this connection is straightforward, to our knowledge it has not previously been used to derive lower bounds for Hamiltonian learning and testing. This connection motivates applying ideas from the literature of quantum query lower bounds to the context of Hamiltonian learning and testing, where there were no lower bounds going beyond the Heisenberg limit. In particular, the tool we borrow is the continuous-time variant of the adversary method~\cite{ambainis2000quantum,hoyer2007negative,brandeho2015universal}. Rather than invoking a general adversary lower-bound theorem, we construct the progress function and bound its derivative directly for our input distribution.

At a high level, the continuous-time access model for diagonal Hamiltonians $H=\sum_{x\in\{0,1\}^n}h_x\ket{x}\bra{x}$ works as follows. Let ${\cal Q}$ be the $n$-qubit register on which $H$ acts, let ${\cal W}$ contain the algorithm's workspace registers, and write ${\cal A} = {\cal W} \otimes {\cal Q}$. Then the input-dependent Hamiltonian, also known as the \textit{query Hamiltonian}, on these registers is
\begin{equation*}
    H_Q(h)=I_{\mathcal W}\otimes H=\sum_x h_x\bigl(I_{\mathcal W}\otimes\proj{x}\bigr).
\end{equation*}
The algorithm chooses an input-independent Hamiltonian $K(t)$, also known as the \textit{driver Hamiltonian}. The algorithm also chooses the \emph{query intensity} $\lambda(t)$, and the resulting evolution is
\begin{equation*}
\id\frac{\dd}{\dd t}\ket{\psi_H(t)}=\bigl(K(t)+\lambda(t)H_Q(h)\bigr)\ket{\psi_H(t)}.
\end{equation*}
For example, taking $K=0$ and $\lambda=1$ for a duration $t_i$ implements the query $e^{-\id t_iH}$. Taking $\lambda=0$ allows the algorithm to perform an $H$-independent operation, such as one of the unitaries $V_i$ in an experiment of \cref{fig:access-model}. A slight extension of this access model, described in \cref{sec:conttimemodel}, also accommodates adaptive choices in between the experiments. The algorithm runs from $t=0$ to $t=t_f$, and its total evolution time is given by $T=\int_0^{t_f}\lambda (t)\dd t.$

As we detail in \cref{subsec:accessmodel}, the continuous-time access model subsumes the experiment model of~\cref{fig:access-model}, so lower bounds in the continuous-time access model carry over to the model of our interest. Working in the continuous-time access model is mathematically more convenient, because the adversary method allows us to define a progress function $\mathcal P(t)$ that 
\begin{enumerate}[label=(\roman*)]
    \item increases by a constant whenever the algorithm distinguishes the zero Hamiltonian from our random ensemble with constant advantage; \label{item:intro1}
    \item satisfies $\frac{\dd}{\dd t}\mathcal{P}(t)=O(\eps^2\lambda (t))$.\label{item:intro2}
\end{enumerate}
\Cref{item:intro1} implies that if an algorithm is able to distinguish the zero Hamiltonian from the random one, then
\begin{equation*}
    \mathcal P(t_f)-\mathcal P(0)=\Omega(1).
\end{equation*}
On the other hand, by \Cref{item:intro2}, we have that 
\begin{equation*}
    \mathcal P(t_f)-\mathcal P(0)=O(\eps^2)\int_0^{t_f}\lambda(t)\dd t=O(\eps^2T).
\end{equation*}
Combining these two bounds yields the desired $T=\Omega(1/\eps^2).$ 

The construction of the progress function is rather technical, and tailored towards distinguishing the zero Hamiltonian from our chosen random ensemble. In particular, it requires \emph{purifying} the randomness of the ensemble into an extra input register, as is standard in adversary lower bounds. Once there, the goal is to find a progress function that quantifies the distinguishability of the purified evolution of the zero Hamiltonian $\ket{\Psi_Z(t)}$ and the one of the random Hamiltonian, that is \emph{far} from zero, $\ket{\Psi_F(t)}$. A natural candidate for $\mathcal P(t)$ would be $1-\mathrm{Re}(\langle \Psi_F(t)|\Psi_Z(t)\rangle)$, since it is zero initially, and at time $t_f$ the algorithm must have separated both purified evolutions, so then $1-\mathrm{Re}(\langle \Psi_F(t_f)|\Psi_Z(t_f)\rangle)=\Omega(1)$. However, directly bounding its derivative does not yield the appropriate bound. We therefore introduce a \emph{correction} term of the kind $\frac{1}{4}\bra{\Psi_F(t)}(I_{\cal A} \otimes G)\ket{\Psi_F(t)}$, where $G$ is a carefully chosen operator. For readers familiar with adversary lower bounds, the operator $G$ plays essentially the same role as choosing the adversary operator $\Gamma$. For a precise definition of our ensemble of random Hamiltonians, we refer to \cref{sec:hardinst}, and for a precise construction of $G$, we refer to \cref{sec:progress-rate}.

\subsection{Future directions}
Our work gives lower bounds for Hamiltonian locality testing, Hamiltonian certification and amplitude estimation, under mild assumptions of the kind $(1/\eps)=2^{cn}$, for some constant $c$. Some restriction of this kind is unavoidable: an $\Omega(1/\eps^2)$ lower bound cannot hold for arbitrarily small $\eps$ at fixed $n$. A natural question is how far our lower bounds extend into
the \emph{high-precision} regime. For each problem, which is the largest constant $c>0$ such that the $\Omega(1/\eps^2)$ lower bound remains valid when $\eps=2^{-cn}$? 

Our work also opens a new route to give lower bounds in Hamiltonian property testing. We wonder if the adversary method can also be used in other problems, such as testing whether a Hamiltonian is sparse in the Pauli basis, studied in \cite{arunachalam2025testing,sinha2025improvedhamiltonianlearningsparsity,bluhm2024hamiltonian}. In particular, we ask: Can the adversary method lead to an optimal lower bound for Hamiltonian sparsity testing?

\paragraph{Note added.} During the preparation of this manuscript we became aware of a related work by Allen, Caro, Chen, Gur and Lowe~\cite{allen2026power}. Their main results coincide with our lower bounds $\Omega(1/\eps^2)$ in the low-precision regime~\cref{res:hardnesscertification,res:hardnesslocality,res:hardnessofAE}. However, their proof technique is different, and, in particular, does not follow the adversary method. They also obtained extra results, such as lower bounds in the high-precision regime $\eps<2^{n/2}$. Only after they shared their draft with us, we realized that our work also implies the same lower bounds in the high-precision regime. 
\section{Preliminaries}
\subsection{Notation}
We write $[n]:=\{1,\ldots,n\}$ and $\R_+:=[0,\infty)$. Unless specified otherwise, real intervals carry Lebesgue measure and $\{0,1\}^n$ carries the uniform probability measure. In particular, $\E_x [f(x)]=2^{-n}\sum_{x\in\{0,1\}^n}f(x)$.

\paragraph{Operators.}
For Hilbert spaces $\cH$ and $\mathcal K$, let $B(\cH,\mathcal K)$ be the space of bounded linear operators from $\cH$ to $\mathcal K$, and write $B(\cH):=B(\cH,\cH)$. For finite-dimensional $\cH$, we also use $L(\cH)=B(\cH)$. We use $\norm{\cdot}$ for Hilbert-space norms and $\norm{\cdot}_{\op}$ for operator norms. For self-adjoint operators, $A\preceq B$ means that $B-A$ is positive semidefinite. All projections are orthogonal, and $[A,B]:=AB-BA$ denotes the commutator. Operators assigned to one tensor factor act as the identity on the others. The notation $\ket{0}_{\cH}$ denotes a fixed unit vector, whereas $0_{\cH}:\mathcal H\to\mathcal H$ denotes the zero operator.

\paragraph{Function spaces and measurability.}
We use the usual spaces $L^p(S,\mu)$, $1\leq p\leq\infty$, with functions identified up to equality almost everywhere, and omit $\mu$ for Lebesgue measure. We identify $f\in L^2(S,\mu)$ with the vector $\ket{f}$, and write $L^p(S;\cH)$ for Hilbert-space-valued functions.

A function into a Banach space is \emph{strongly measurable} if it is an almost-everywhere limit, in norm, of measurable simple functions. For operator-valued functions taking values in $B(\cH)$, the norm here is the operator norm. In finite dimensions this is equivalent to measurability of the matrix entries. If $F$ is strongly measurable and $\int\norm{F(t)}\dd t<\infty$, its integral is understood in the Bochner sense. In particular, if $v(t)=v(0)+\int_0^t w(s)\dd s$ with $w\in L^1([0,t_f];\cH)$, then $v$ is absolutely continuous and $\dot v(t)=w(t)$ almost everywhere.

\paragraph{Fourier analysis.}
For $S\subseteq[n]$, define the Boolean character $\chi_S(x):=(-1)^{\sum_{j\in S}x_j}$. These characters form an orthonormal basis under the uniform measure. For $f:\{0,1\}^n\to\C$, define $\widehat f(S):=\E_x\bigl[f(x)\chi_S(x)\bigr]$. Then,
\begin{equation*}
    f(x)=\sum_{S\subseteq[n]}\widehat f(S)\chi_S(x).
    \label{eq:prelim-cube-fourier}
\end{equation*}
Parseval's identity gives
\begin{equation*}
    \E_x|f(x)|^2=\sum_{S\subseteq[n]}|\widehat f(S)|^2.
    \label{eq:prelim-cube-parseval}
\end{equation*}
On $\R$, we use the Fourier transform convention
\begin{equation*}
    \widehat f(t):=\int_{\R}f(\omega)e^{-\id\omega t}\dd\omega
    \label{eq:prelim-real-fourier}
\end{equation*}
for $f\in L^1(\R)$. The transform extends to $L^2(\R)$, where Parseval's identity reads
\begin{equation*}
    \int_{\R}|\widehat f(t)|^2\dd t =2\pi\int_{\R}|f(\omega)|^2\dd\omega.
    \label{eq:prelim-plancherel}
\end{equation*}
If $f$ and $\widehat f$ are integrable, Fourier inversion gives $f(\omega)=(2\pi)^{-1}\int_{\R}\widehat f(t)e^{\id\omega t}\dd t$ almost everywhere. Functions defined on a real interval are extended by zero when taking their Fourier transform on $\R$. The convolution of two functions is
\begin{equation*}
    (f*g)(t):=\int_{\R}f(t-s)g(s)\dd s,
\end{equation*}
whenever this integral exists. For functions on $\R_+$, we use their zero extensions to $\R$; consequently, $(f*g)(t)=\int_0^t f(t-s)g(s)\dd s$ for $t\geq0$.

\subsection{Hamiltonians}\label{subsec:Hamiltonians}

An $n$-qubit Hamiltonian is a Hermitian operator on $(\C^2)^{\otimes n}$. The unknown Hamiltonians in this work satisfy $\norm{H}_{\op}\leq1$. We also use bounded self-adjoint operators on auxiliary, possibly infinite-dimensional, Hilbert spaces. The evolution under a time-independent Hamiltonian $H$ is given by $U_H(t):=e^{-\id tH}$. Thus a state $\rho$ evolves to $U_H(t)\rho U_H(t)^\dagger$ after time $t$.

\paragraph{The Pauli basis and locality.}
Let $\Pauli:=\{I,X,Y,Z\}^{\otimes n}$ be the $n$-qubit Pauli basis. For $P=P_1\otimes\cdots\otimes P_n\in\Pauli$, its support is $\supp(P):=\{j\in[n]:P_j\neq I\}$ and its weight is $|P|:=|\supp(P)|$. The Pauli operators form an orthonormal basis for the normalized Hilbert--Schmidt inner product $\ip{A}{B}_2:=2^{-n}\Tr(A^\dagger B)$. For an $n$-qubit operator $A$, define $\widehat A(P):=2^{-n}\Tr(PA)$. Its Pauli expansion is
\begin{equation*}
    A=\sum_{P\in\Pauli}\widehat A(P)P.
    \label{eq:prelim-pauli-expansion}
\end{equation*}
We use the \emph{normalized Frobenius norm} $\norm{A}_2:=\sqrt{2^{-n}\Tr(A^\dagger A)}$. Parseval's identity reads
\begin{equation*}
    \norm{A}_2^2=\sum_{P\in\Pauli}|\widehat A(P)|^2.
    \label{eq:prelim-pauli-parseval}
\end{equation*}
This normalization gives $\norm{A}_2\leq\norm{A}_{\op}$ and $\norm{I}_2=1$. The coefficients $\widehat H(P)$ of a Hamiltonian are real.

For $0\leq k\leq n$, a Hamiltonian is \emph{$k$-local} if it is a sum of Hermitian terms, each acting nontrivially on at most $k$ qubits. Equivalently, $\widehat H(P)=0$ whenever $|P|>k$. For an arbitrary operator $A$, its truncation to weight at most $k$ is
\begin{equation*}
    A_{\leq k}:=\sum_{\substack{P\in\Pauli\\|P|\leq k}}
        \widehat A(P)P.
\end{equation*}
Its component of weight greater than $k$ is
\begin{equation*}
    A_{>k}:=\sum_{\substack{P\in\Pauli\\|P|>k}}\widehat A(P)P.
\end{equation*}
These definitions also apply to time-evolution operators, which need not be Hermitian.

\paragraph{Distance from locality.}
We measure distance using the normalized Frobenius norm:
\begin{equation*}
    \dist(H,k\text{-local})
    :=\inf\bigl\{\norm{H-L}_2:
          L=L^\dagger\text{ is }k\text{-local}\bigr\}.
\end{equation*}
Orthogonality of the Pauli basis gives, for every $k$-local $L$,
\begin{equation*}
    \norm{H-L}_2^2
       =\norm{H_{>k}}_2^2+\norm{H_{\leq k}-L}_2^2.
\end{equation*}
Hence $H_{\leq k}$ is a closest $k$-local Hamiltonian, with $\dist(H,k\text{-local})=\norm{H_{>k}}_2$. Equivalently,
\begin{equation*}
    \dist(H,k\text{-local})^2 =\sum_{\substack{P\in\Pauli\\|P|>k}}|\widehat H(P)|^2.
    \label{eq:prelim-pauli-distance}
\end{equation*}

For a diagonal Hamiltonian $H=\sum_{x\in\{0,1\}^n}h_x\proj{x}$, write $h:\{0,1\}^n\to\R$ for the function $h(x)=h_x$. If $Z_S$ denotes the Pauli operator with $Z$ on $S\subseteq[n]$ and $I$ elsewhere, then $Z_S\ket{x}=\chi_S(x)\ket{x}$ and
\begin{equation*}
    H=\sum_{S\subseteq[n]}\widehat h(S)Z_S.
\end{equation*}
Here $\widehat h(S)$ is the Boolean Fourier coefficient defined above. All Pauli coefficients containing an $X$ or a $Y$ vanish. Therefore
\begin{equation*}
    \dist(H,k\text{-local})^2
       =\sum_{\substack{S\subseteq[n]\\|S|>k}}|\widehat h(S)|^2.
    \label{eq:prelim-diagonal-distance}
\end{equation*}
By Parseval's identity, this can also be written as
\begin{equation*}
    \dist(H,k\text{-local})^2
       =2^{-n}\sum_x|h_x|^2
         -\sum_{\substack{S\subseteq[n]\\|S|\leq k}}
             |\widehat h(S)|^2.
    \label{eq:prelim-diagonal-distance-parseval}
\end{equation*}
In particular, $\norm{H}_{\op}=\max_x|h_x|$ and $\norm{H}_2^2=\E_x|h_x|^2$.

\paragraph{Time-dependent evolution.}
Let $F:[0,t_f]\to B(\cH)$ be strongly measurable, self-adjoint almost everywhere, and satisfy $\ell:=\int_0^{t_f}\norm{F(s)}_{\op}\dd s<\infty$. The propagator starting at time zero is given by the \emph{Dyson expansion}
\begin{equation}
    U_F(t)=I_{\cH}+\sum_{m=1}^{\infty}(-\id)^m
       \int_{0<s_m<\cdots<s_1<t}
       F(s_1)\cdots F(s_m)\dd s_1\cdots\dd s_m.
    \label{eq:prelim-dyson}
\end{equation}
The $m$th term has operator norm at most $\ell^m/m!$, so the series converges absolutely and uniformly for $0\leq t\leq t_f$. It defines the unique absolutely continuous solution of
\begin{equation*}
    U_F(t)=I_{\cH}-\id\int_0^t F(s)U_F(s)\dd s,
\end{equation*}
and satisfies $\id\dot U_F(t)=F(t)U_F(t)$ almost everywhere. Self-adjointness of $F$ makes $U_F(t)$ unitary. These statements also hold on infinite-dimensional Hilbert spaces under the boundedness and integrability assumptions above. For a fixed bounded observable $A$ and $\ket{\psi(t)}=U_F(t)\ket{\psi(0)}$, the corresponding identity is
\begin{equation}
    \frac{\dd}{\dd t}\bra{\psi(t)}A\ket{\psi(t)}=\bra{\psi(t)}\id[F(t),A]\ket{\psi(t)}
    \label{eq:prelim-commutator-evolution}
\end{equation}
almost everywhere.

\subsection{Problem and access model}\label{subsec:accessmodel}

We consider an access model physically motivated by the role of Hamiltonians in quantum dynamics. Due to the Schrödinger equation, if the interactions of a quantum system are governed by a Hamiltonian $H,$ the states of the system evolve according to the time evolution operator $U(t)=U_H(t)=e^{-itH}$. This means that if the (mixed) state describing the system at time $0$ is $\rho,$ then the state at time $t$ is $\rho(t)=U(t)\rho U^{\dagger}(t)$. Thus, a natural access model for Hamiltonians is to perform \emph{experiments} of the following kind (see \cref{fig:access-model}): prepare a state $\rho,$ apply $U(t_1)$ (this is making a query to $U(t_1)$, which in a lab can be implemented by letting the system evolve for time $t_1$), apply a unitary operator $V_1$ independent of $H$, query $U(t_2)$, apply a unitary operator $V_2$ independent of $H$, query $U(t_3)$,$\dots$ and finally measure. We allow the use (a finite number of) ancilla qubits and to make adaptive choices, in between experiments, of both the $V_i$'s and the times of evolution. We also assume the number of queries to the time evolution operator and the number of experiments to be finite. Notably, we only allow forward time evolution, i.e. we only consider $t_i\geq 0$. This access model was previously considered in~\cite{huang2023heisenberg,bakshi2024structure,arunachalam2025testing,kallaugher2025hamiltonian,gao2025quantum,sinha2025improvedhamiltonianlearningsparsity,bluhm2026certifying,chen2025lower,shin2026emergent,sinha2025improvedhamiltonianlearningsparsity}.

Now, we are ready to state the problems of Hamiltonian locality testing and certification. 
\begin{definition}[Hamiltonian locality testing]
    Let $k,n\in\N,$ and $0<\eps<1$. Assume that we have time evolution access to an unknown $n$-qubit Hamiltonian $H$ via experiments like those of \cref{fig:access-model}. Assume that $H$ satisfies $\norm{H}_\op\leq 1$. The problem of \emph{Hamiltonian locality testing} is to decide whether $H$ is $k$-local or satisfies $\dist(H,k$-local)$>\eps$, under the promise that one of the two cases is satisfied.
\end{definition}

\begin{definition}[Hamiltonian certification]
    Let $n\in\N$ and $0<\eps<1$. Let $H_0$ be an $n$-qubit Hamiltonian that is known a priori, and satisfies $\norm{H_0}_{\op}\leq 1$ and $\Tr[H_0]=0$. Assume that we have access to an unknown $n$-qubit Hamiltonian $H$ via experiments like those of \cref{fig:access-model}. Assume that $H$ satisfies $\norm{H}_\op\leq 1$ and $\Tr[H]=0$. The problem of \emph{certifying} $H_0$ is to decide whether $H$ equals $H_0$ or $\norm{H-H_0}_2>\eps$, under the promise that one of the two cases is satisfied.\footnote{We impose that $\Tr[H]=\Tr[H_0]$ because if $H$ and $H_0$ differ by a scalar multiple of the identity, they might be far, but they would generate the same time evolution (up to a global phase).}
\end{definition}

In this work, we only care about one measure of complexity, the \emph{total evolution time}, which is the sum of the query durations across all experiments. As we allow the algorithm to make adaptive decisions, the total evolution time might vary depending on the input Hamiltonian; we consider the worst-case total evolution time $T$, i.e. the largest total evolution time of the algorithm over all possible input Hamiltonians and measurement outcomes.

\subsubsection{Continuous-time queries}\label{sec:conttimemodel}
The experiments described above give a natural way to access an unknown Hamiltonian $H$. To prove our lower bound, we use a continuous-time query model, following the Hamiltonian-oracle framework~\cite{farhi1998analog,mochon2007hamiltonian} and its adversary-method treatment~\cite{brandeho2015universal}. We show that every experiment of the form of \cref{fig:access-model} can be described in this model without increasing its total evolution time and hence a lower bound in this model applies to the experiments that we wish to study. 

Before formally introducing the continuous-time access model, we give a high-level motivation of it. We only need this description for diagonal Hamiltonians $H=\sum_{x\in\{0,1\}^n}h_x\ket{x}\bra{x}$, since these are the inputs considered in our proof. Let ${\cal Q}$ be the $n$-qubit register on which $H$ acts, let ${\cal W}$ contain the algorithm's workspace registers, and write ${\cal A} = {\cal W} \otimes {\cal Q}$. Then the input-dependent Hamiltonian, also known as the \textit{query Hamiltonian}, on these registers is
\begin{equation}\label{eq:query-H}
    H_Q(h)=I_{\mathcal W}\otimes H=\sum_x h_x\bigl(I_{\mathcal W}\otimes\proj{x}\bigr).
\end{equation}
The algorithm chooses an input-independent Hamiltonian $K(t)$, also known as the \textit{driver Hamiltonian}, to implement its input-independent operations. The algorithm also picks the \emph{query intensity} $\lambda(t)$, and the resulting evolution is
\begin{equation}\label{eq:standard-continuous}
\id\frac{\dd}{\dd t}\ket{\psi_H(t)}=\bigl(K(t)+\lambda(t)H_Q(h)\bigr)\ket{\psi_H(t)}.
\end{equation}
For example, taking $K=0$ and $\lambda=1$ for a duration $t_i$ implements the query $e^{-\id t_iH}$. Taking $\lambda=0$ allows the algorithm to perform an $H$-independent operation, such as one of the unitaries $V_i$ in an experiment. As before, these operations are uncharged: the cost is the accumulated query time $\int_0^{t_f}\lambda(t)\dd t$, rather than the duration $t_f$ of the whole description. We require $\lambda\geq0$ to retain the restriction to forward evolution.

Lastly, in an adaptive algorithm, different measurement outcomes may lead to queries of different durations. When we keep these outcomes in a workspace register, we must be able to let some branches continue querying while others wait. We therefore replace $I_{\mathcal W}\otimes\proj{x}$ in \cref{eq:query-H} by projections $P_x(t)$ that specify where the query is still active. These projections may leave part of the space inactive, which is expressed by $\sum_xP_x(t)\preceq I_{\cal A}$. The projections in the definition can be understood directly from the experiments. During an ordinary query, we take $P_x=I_{\mathcal W}\otimes\proj{x}$ and recover \cref{eq:standard-continuous}. If only some branches are querying, we instead take $P_x(t)=C(t)\otimes\proj{x}$, where $C(t)$ projects onto their measurement records in the workspace. The other branches then undergo no oracle evolution. 

Now, we are ready to formally introduce our continuous-time access model. 

\begin{tcolorbox}[
    title=Continuous-time access model for diagonal Hamiltonians,
    colback=white,
    colframe=black]
    
    Let $\mathcal Q$ be a $2^n$-dimensional complex Hilbert space, on which the unknown diagonal Hamiltonian $H=\sum_{x\in\{0,1\}^n}h_x\ket{x}\bra{x}$ acts. Let $\mathcal W$ be a complex finite-dimensional Hilbert space. Let $\cA=\mathcal W\otimes \mathcal{Q}$ and fix a normalized basis vector $\ket{0}_{\cA}$ as the initial state. A continuous-time algorithm is specified by the following parameters, all independent of the unknown $H$:
    \begin{enumerate}
        \item A \emph{final time} $t_f\in [0,\infty)$. 
        \item A strongly measurable function $K:[0,\infty)\to L(\cA)$ such that $K(t)=K(t)^\dagger$ almost everywhere, and $\int_0^{t_f}\norm{K(t)}_{\op}\dd t< \infty$. We say that $K$ is a \emph{driver Hamiltonian}.
        
        \item A \emph{query intensity} $\lambda\in L^1([0,t_f])$ that takes nonnegative values.
        \item For every $x\in\{0,1\}^n$, a strongly measurable projection-valued function $P_x:[0,\infty)\to L(\cA)$, such that for almost every $t$, $P_x(t)P_y(t)=0$ for $x\neq y$, and $\sum_xP_x(t)\preceq I_{\cA}$.
        \item A family of mutually orthogonal projectors $\{P_o\}_{o\in\mathcal{O}}$ on $\cA$ such that $\sum_{o\in\mathcal{O}} P_o=I_{\cA}.$
    \end{enumerate}
    When the input to the algorithm is the Hamiltonian $H=\sum_{x\in\{0,1\}^n}h_x\ket{x}\bra{x}$, the algorithm prepares the state $\ket{\psi_H(t_f)},$ where for $t\in [0,t_f]$ the states $\ket{\psi_H(t)}$ are the solution to  
    \begin{equation}\label{eq:continuous-integral}
  \ket{\psi_H(t)}
  =\ket{0}_\cA
  -\id\int_0^{t}
  \left(K(s)+\lambda(s)\sum_x h_xP_x(s)\right)
  \ket{\psi_H(s)}\dd s.
\end{equation}
    The algorithm performs the measurement $\{P_o\}_{o\in\mathcal O}$, and post-processes the outcome of this measurement to decide the final output. The total evolution time of the algorithm is $$T:=\int_0^{t_f}\lambda(s)\dd s.$$
\end{tcolorbox}

It remains to address two details: that the integral equation defines a valid evolution, and that the adaptive experiments admit such a description with the same worst-case time bound. The fact that the continuous time model is well-defined follows from the fact that the Dyson expansion is well-defined, as we argued in \cref{subsec:Hamiltonians}.

Now, we show that the continuous-time model subsumes the model of~\cref{fig:access-model}. Consider a protocol that performs experiments as in~\cref{fig:access-model}. We show that given any $\eps>0,$ the continuous time access model can simulate the protocol up to error $\eps$ without increasing the total evolution time. As the algorithm performs a finite number of experiments and the space on which it acts is also finite, the algorithm has a finite number of rounds of adaptivity in between experiments. In particular, there is also a finite number of times at which the algorithm queries the time evolution operation. Let $Q\in\N$ be the maximum number of queries that the protocol performs. Let $\delta:=\eps/Q$. Given one evolution time $t_i$ at which the algorithm queries $U_H$, let $m_i:=\lfloor t_i/\delta\rfloor$. Then, $|t_i-m_i\delta|\leq \eps/Q$ and $m_i\delta\leq t_i$. Substituting the query to $U_H(t_i)$ by $m_i$ queries to $U_H(\delta)$  one incurs into at most $\eps/(2Q)$ error in diamond norm. As a consequence, one can simulate the protocol up to error $\eps$ making only queries to $U_H(\delta)$ without increasing the total evolution time. 
    
It remains to exactly simulate with continuous queries a protocol whose queries all have duration~$\delta$. Each branch makes at most $N:=\lfloor T/\delta\rfloor$ such queries. By the deferred measurement principle \cite[Section~4.4]{nielsen_chuang}, we retain the outcomes coherently and may control subsequent operations on them. At each of the at most $N$ query steps, unfinished branches perform their next query, while completed branches remain inactive. In the continuous-time model, the projections $P_x(t)$ ensure that the oracle acts only on branches performing a query, leaving completed branches unchanged. Between the query steps, we set $\lambda=0$ and use the driver Hamiltonian to implement the input-independent operations. This simulation is exact and has total evolution time $N\delta\leq T$.

\subsection{Useful inequalities}\label{sec:useful-inequalities}

We recall several standard inequalities used throughout the paper.

\begin{lemma}[Markov's inequality]
\label{lem:markov}
Let $X$ be a nonnegative random variable with finite expectation.
For every $t>0$,
\[
    \Pr(X \ge t) \le \frac{\mathbb{E}[X]}{t}.
\]
\end{lemma}

\begin{lemma}[Chebyshev's inequality]
\label{lem:chebyshev}
Let $X$ be a real-valued random variable with mean
$\mu=\mathbb{E}[X]$ and finite variance
$\sigma^2=\mathbb{E}[(X-\mu)^2]$.
For every $t>0$,
\[
    \Pr\bigl(|X-\mu| \ge t\bigr)
    \le \frac{\sigma^2}{t^2}.
\]
\end{lemma}

\begin{lemma}[Young's convolution inequality]
\label{lem:young-convolution}
Let $p,q,r\in[1,\infty]$ satisfy $1+\frac{1}{r}=\frac{1}{p}+\frac{1}{q}$. Let $f\in L^p(\mathbb{R})$ and $g\in L^q(\mathbb{R})$. Then, 
$\lVert f*g\rVert_{L^r(\R)} \le \lVert f\rVert_{L^p(\R)}\, \lVert g\rVert_{L^q(\R)}$.
\end{lemma}

\section{Distinguishing the zero Hamiltonian and a random one is hard}
In this section, we prove our main technical result, \cref{thm:mainresult}. Our proof uses the continuous-time adversary method. We develop the argument directly for distinguishing the zero Hamiltonian from a certain ensemble of random Hamiltonians, introducing the adversary operator and progress function as they are needed, rather than presenting the general framework. This keeps the notation specific to our problem and makes it easier to see the role of each construction in the proof. For background on the continuous-time adversary method, we refer the interested reader to Brandeho--Roland~\cite{brandeho2015universal}; for broader extensions to other oracle models, see Belovs~\cite{belovs2015variations}.

We first construct the family of random Hamiltonians, and in \cref{sec:corollaries}, we show that a Hamiltonian drawn from this family is far from local and from the zero Hamiltonian with high probability. We then construct a progress function that is initially zero and we show that it must increase by a constant for any successful algorithm. By bounding its increase rate per unit of evolution time, we show that any algorithm requires $\Omega(1/\eps^2)$ total evolution time to distinguish our two cases.

\subsection{The random Hamiltonians}\label{sec:hardinst}

The hard instance will be a family of random Hamiltonians that are diagonal in the computational basis. More precisely, we consider Hamiltonians
\begin{equation*}
    H=\sum_{x\in\{0,1\}^n}h_x\ket{x}\bra{x}
\end{equation*}
where every $h_x$ is independently drawn from certain probability distribution $\nu_\eps$ taking values in $[-1,1]$ that has several useful properties. We  first define the distribution, and then prove several of its properties. 

Defining the distribution requires some notation. We consider functions $r$ of $L^2([-1/2,1/2])$, which we also treat as elements of the Hilbert space, in which case we denote them by $\ket{r}$. We introduce the self-adjoint operator $M$ acting on $L^2([-1/2,1/2])$ that maps $ r\in L^2([-1/2,1/2])$ to the function $w\to wr(w)$, i.e. $(Mr)(w)=wr(w).$
Note that $\norm{M}_\op\leq 1/2$. We also consider the \textit{bump function} $b\in L^2([-1/2,1/2])$ defined by 
\begin{equation}
    b(w):=\dfrac{\widetilde b(w)}{\norm{\widetilde b}_{L^2([-1/2,1/2])}},\quad\text{where}\quad  \widetilde b(\omega):=
  \begin{cases}
    \exp\!\left(-\dfrac{1}{1-16\omega^2}\right),& |\omega|<1/4,\\[1.2ex]
    0,& |\omega|\geq 1/4.
  \end{cases} \label{eq:bumpfunction} 
\end{equation}
Given $0 < \eps < 1/2,$ we define the operator $J_\eps$ acting on $\text{span}_\C\{\ket{\perp}\}\oplus L^2([-1/2,1/2])$,\footnote{$\ket{\perp}$ just denotes a unit vector of $\ell_2(\mathbb C)$.} as 
\begin{equation}\label{eq:blockJ}
    J_\eps:=\begin{pmatrix}
        0 & \eps\bra b\\ \eps\ket b & M
    \end{pmatrix}
\end{equation}
Note that $\norm{J_\eps}_\op\leq 1,$ because 
\begin{equation*}
    \norm{J_\eps}_{\op}\leq \norm{\begin{pmatrix}
        0 & \eps\bra b\\ \eps\ket b & 0
    \end{pmatrix} }_\op+\norm{ M}_\op\leq \eps+\dfrac{1}{2}\leq 1.
\end{equation*}
Now, we are ready to define the distribution $\nu_\eps$.
\begin{definition}[The hard distribution]\label{def:harddist}
    We define $\nu_\eps$ as the unique Borel probability measure on $[-1,1]$ such that 
    \begin{equation*}
        \bra \perp f(J_\eps)\ket \perp=\int_{-1}^1f(h)\dd\nu_\eps(h),
    \end{equation*}
    for every continuous function $f:[-1,1]\to\mathbb C.$
\end{definition}
The distribution $\nu_\eps$ is well defined thanks to the facts that $\langle \perp|\perp\rangle=1$ and $\norm{J_{\eps}}_\op\leq 1$, and the Riesz-Markov representation theorem (see e.g.~\cite{ii1975methods}). Now, we state and prove the properties of the distribution that allow us to ensure that the random Hamiltonian ensemble is far from local and from the zero Hamiltonian with constant probability.
 
\begin{lemma}[Properties of the hard distribution that ensure farness]\label{lem:propertiesharddist}
    Let $1/2 > \eps>0$. The probability distribution $\nu_\eps$ from \cref{def:harddist} satisfies
    \begin{enumerate}[label=(\roman*)]
        \item \label{item:1} $\mathbb E_{h\sim \nu_\eps}[h]=0,$
        \item \label{item:2} $\mathbb E_{h\sim \nu_\eps}[h^2]=\eps^2$.  
    \end{enumerate}
\end{lemma}
\begin{proof}
    \cref{item:1} follows from the fact that $\bra \perp J_\eps\ket \perp=0$. \cref{item:2} follows from $\bra \perp J_\eps^2\ket\perp= \eps^2\langle b|b\rangle$ and $\norm{b}_{L^2([-1/2,1/2])}=1.$
\end{proof}

As we will see in \cref{sec:corollaries}, the properties from \cref{lem:propertiesharddist} suffice to establish farness, but not for the hardness. For example, if one picked $h_x$ as a uniformly random variable taking values in $\{-\eps,\eps\},$ then the random Hamiltonian will still satisfy the properties of \cref{lem:propertiesharddist}, but it will be distinguishable from the zero Hamiltonian with evolution time $O(1/\eps)$. Indeed, note that in the case $|h_x|=\eps$, we have that $U_H(\pi/\eps)=-I_{\C^{2^n}}$, while $U_0(\pi/\eps)=I_{\C^{2^n}}$. Then, the following protocol, that assumes access to the controlled time evolution, distinguishes both cases with probability 1: prepare $\ket {+}\otimes \ket{0}_{\C^{2^n}}$, apply $cU_H(\pi/\eps)$, and measure the first qubit in the $\{\ket +,\ket -\}$ basis. In the case that $U_H(\pi/\eps)=I_{\C^{2^n}}$, the output is $\ket +$, while in the case $U_H(\pi/\eps)=-I_{\C^{2^n}}$, the output is $\ket -$. Finally, the protocol can be simulated just with non-controlled time evolution up to arbitrary precision without incurring in any total time evolution overhead, thanks to the controlization techniques of~\cite{odake2025universal}. Our choice of distribution provides the following additional estimates that allow us to establish hardness of distinguishing the zero Hamiltonian from the random ensemble. We will not state these estimates in terms of the distribution itself, but instead state them in terms of $M$ and $\ket {b}$, which determine the distribution $\nu_\eps.$

\begin{lemma}[Properties of the hard distribution used by the adversary]\label{lem:propertieshardextra}
Let $b\in L^2([-1/2,1/2])$ be the bump function defined in~\cref{eq:bumpfunction}. Then, there are constants $C,C_{\rm obs}<\infty$, such that
\begin{align}
\int_0^\infty\abs{\bra{b}e^{-\id M\tau}\ket{b}}\dd \tau&\leq C,\label{eq:hard-integrable}\\
\int_0^\infty\abs{\bra{b}e^{-\id M\tau}\ket{r}}^2\dd \tau&\leq C_{\rm obs}\norm{\ket{r}}^2\label{eq:hard-observation}
\end{align}
for every $\ket{r}\in L^2([-1/2,1/2])$. 
\end{lemma}
\begin{proof}
We first prove \cref{eq:hard-integrable}. We start by dividing the integral in two summands, 
\begin{equation*}
    \int_0^\infty\abs{\bra{b}e^{-\id M\tau}\ket{b}}\dd \tau=\int_0^1\abs{\bra{b}e^{-\id M\tau}\ket{b}}\dd \tau+\int_1^\infty\abs{\bra{b}e^{-\id M\tau}\ket{b}}\dd \tau
\end{equation*}
Since $\norm{b}_{L^2([-1/2,1/2])}=1$, we have $\abs{\bra{b}e^{-\id M\tau}\ket{b}}\leq1$, so the first summand is bounded by 1. To bound the second summand, we note that by integrating by parts twice and using that $b(\pm 1/2)=0$, we obtain
\[
\bra{b}e^{-\id M\tau}\ket{b}=\int_{-1/2}^{1/2}\abs{b(\omega)}^2e^{-\id\omega \tau}\dd\omega=\frac{1}{(\id \tau)^2}\int_{-1/2}^{1/2}(\abs{b}^2)''(\omega)e^{-\id\omega \tau}\dd\omega \qquad(\tau\neq0).
\]
Thus, $\abs{\bra{b}e^{-\id M\tau}\ket{b}}\leq \tau^{-2}\norm{(\abs{b}^2)''}_{L^1([-1/2,1/2])}$. As $b$ is smooth, $C'=\norm{(\abs{b}^2)''}_{L^1([-1/2,1/2])}$ is a constant. Thus, we can upper bound the second summand by constant
\begin{equation*}
    \int_1^\infty\abs{\bra{b}e^{-\id M\tau}\ket{b}}\dd \tau\leq C'\int_1^\infty \tau^{-2}d\tau=C',
\end{equation*} which concludes the proof of \cref{eq:hard-integrable}.

For \cref{eq:hard-observation}, note that $$\bra{b}e^{-\id M\tau}\ket{r}=\int_{-1/2}^{1/2}e^{-iw\tau}\bar b(w)r(w)dw=\int_{-\infty}^{\infty}e^{-iw\tau} g(w)dw,$$ where $g(w)=\bar b(w)r(w)$ for $|w|\leq 1/2$ and $g(w)=0$ for $|w|\geq 1/2$. Thus, $\bra{b}e^{-\id M\tau}\ket{r}$ is the Fourier transform of $\overline b r$, extended by zero outside $[-1/2,1/2]$. Since $b\in L^{\infty}([-1/2,1/2])$ and $r\in L^2([-1/2,1/2])$, we have that $g\in L^2(\mathbb R)$, with
\[
\norm{g}_{L^2(\mathbb R)}\leq\norm{b}_{L^\infty([-1/2,1/2])}\norm{r}_{L^2([-1/2,1/2])}.
\]
It also belongs to $L^1(\mathbb R)$: by Cauchy--Schwarz and the normalization of $b$,
\[
\norm{g}_{L^1(\mathbb R)}\leq\norm{b}_{L^2([-1/2,1/2])}\norm{r}_{L^2([-1/2,1/2])}=\norm{r}_{L^2([-1/2,1/2])}.
\]
The Fourier integral is therefore well-defined, and Parseval's identity gives
\[
\int_0^\infty\abs{\bra{b}e^{-\id M\tau}\ket{r}}^2\dd \tau\leq2\pi\norm{g}_{L^2(\mathbb R)}^2\leq2\pi\norm{b}_{L^\infty([-1/2,1/2])}^2\norm{r}_{L^2([-1/2,1/2])}^2,
\]
which proves the bound with $C_{\rm obs}=2\pi\norm{b}_{L^\infty([-1/2,1/2])}^2$, which is a constant because $b$ is smooth. 
\end{proof} 

To explain why \cref{lem:propertieshardextra} is useful, we briefly anticipate the analysis used in the adversary argument in \cref{sec:progress-rate}. We represent each random entry $h_x$ by an auxiliary register initialized in $\ket{\perp}$, with $J_{\eps}$ from \cref{eq:blockJ} acting on that register whenever the algorithm queries the entry. The identity $J_{\eps}\ket{\perp}=\eps\ket{b}$ explains why the overlaps with $\ket{b}$ appear in the proof. A term proportional to $\ket{b}$ introduced at time $\sigma$ contributes to the later overlap through $\bra{b}e^{-\id M(\tau-\sigma)}\ket{b}$. \cref{eq:hard-integrable} lets us bound the combined effect of such terms introduced at different times.
\cref{eq:hard-observation} controls the contribution from a vector already present at the start of the interval. Both bounds are independent of the interval's duration. We use them in \cref{sec:progress-rate} to construct the correction that bounds the algorithm's progress by $O(\eps^2)$ per unit of query time.

\begin{remark}[The uniform distribution does not have the properties of \cref{lem:propertieshardextra}]
    As we explained above, if we had picked $\nu_\eps$ to be the uniform distribution on $\{-\eps,\eps\},$ we would not be able to prove a lower bound better than $\Omega(1/\eps)$. In particular, one could not hope to have a result such as \cref{lem:propertieshardextra}, so our proof would not work. That is because to obtain the uniform distribution as the distribution $\dd\nu_\eps'$ such that
    \begin{equation*}
        \bra \perp f( J_\eps')|\perp\rangle=\int_{-1}^1f(x)\dd\nu_\eps'(x),
    \end{equation*}
    for every $f:[-1,1]\to\C$, we would have to pick $J'_\eps$ as 
    \begin{equation*}
        J'_\eps=\begin{pmatrix}
        0 & \eps\bra{ b'}\\
         \eps\ket{b'} & M'
        \end{pmatrix},
    \end{equation*}
    where $b'=1$ and $M'=0$. In that case, we would have 
    \begin{align*}
        \int_0^\infty\abs{\bra{b'}e^{-\id M'\tau}\ket{b'}}\dd \tau&=\int_0^\infty1\dd \tau=\infty,\\
        \int_0^\infty\abs{\bra{b'}e^{-\id M'\tau}\ket{b'}}^2\dd \tau&=\int_0^\infty1\dd \tau=\infty,
    \end{align*}
ruling out both bounds of \cref{lem:propertieshardextra}.
\end{remark}

\subsection{The progress function}

In this subsection, we design a progress function tailored to measure the progress of an algorithm that distinguishes $H=0$ from our family of random Hamiltonians $H=\sum_x h_x\ket{x}\bra{x}$ with $h_x$ being independently drawn from $\nu_\eps$.

\subsubsection{Purifying the randomness} 
In the adversary method, it is standard to model the input distribution by a superposition in an extra register, thereby \emph{purifying} the randomness of the input. This lets us follow the algorithm and input jointly as a single pure state, avoiding the need of the density matrix formalism. The progress function will then be defined acting (only) on this extra register.

Fix an algorithm of the form from \cref{sec:conttimemodel}, so that on input $H=\sum_x h_x\ket{x}\bra{x}$, the algorithm at time $t$ prepares the state $\ket{\psi_H(t)}$ determined via
\begin{equation*}
  \frac{\dd}{\dd t}\ket{\psi_H(t)}
  =
  -\id
  \left(K(t)+\lambda(t)\sum_x h_xP_x(t)\right)
  \ket{\psi_H(t)},\quad\text{and}\quad \ket{\psi_H(0)}=\ket{0}_{\cA}.
  \label{eq:continuous-evolution}
\end{equation*}

We introduce one input register ${\cal I}_x$ for each entry $h_x$ to model its randomness:
\begin{equation}\label{eq:input-register}
\mathcal I:=\bigotimes_{x\in\{0,1\}^n}\mathcal I_x, \qquad \mathcal I_x:=\C\ket{\perp}_x\oplus L^2([-1/2,1/2]), \qquad \ket{\Omega}:=\bigotimes_x\ket{\perp}_x.
\end{equation}

We use $J_\eps^{(x)}$ to denote the operator that acts as $J_\eps$ on register $\mathcal I_x$ and identity on $\mathcal I_y$ for $y\neq x$. Then, define the joint oracle as
\begin{equation*}
    \mathbf J_\eps(t):=\sum_{x\in\{0,1\}^n}P_x(t)\otimes J_\eps^{(x)}.
\end{equation*}
The algorithm has no access to the register ${\cal I}$, except through this joint oracle.

Now, we define the purified time evolutions of the \emph{zero} and \emph{far} instances as the states $\ket{\Psi_{Z}(t)}, \ket{\Psi_{F}(t)}, \in {\cal A} \otimes {\cal I}$ satisfying 
\begin{align}
    \frac{\dd}{\dd t}\ket{\Psi_{Z}(t)}
  &=
  -\id
  \left(K(t)\otimes I_{\mathcal I}\right)
  \ket{\Psi_{Z}(t)},\quad\text{and}\quad \ket{\Psi_{Z}(0)}=\ket{0}_{\cA}\otimes \ket{\Omega},\label{eq:purifiedevol}\\
  \frac{\dd}{\dd t}\ket{\Psi_{F}(t)}
  &=
  -\id
  \left(K(t)\otimes I_{\mathcal I}+\lambda(t)\mathbf J_\eps(t)\right)
  \ket{\Psi_{F}(t)},\quad\text{and}\quad \ket{\Psi_{F}(0)}=\ket{0}_{\cA}\otimes \ket{\Omega},\nonumber
\end{align}
respectively. We check that these joint evolutions are well-defined, since $\mathcal I$ is infinite-dimensional. Their Hamiltonians are finite sums of fixed bounded self-adjoint operators with measurable real coefficients. Moreover, orthogonality of the $P_x(t)$ gives $\norm{\mathbf J_\eps(t)}_{\op}\leq1$, so their norms are bounded by the integrable function $\norm{K(t)}_{\op}+\lambda(t)$. The existence and uniqueness argument from \cref{sec:conttimemodel} therefore applies here as well.

The following lemma formally shows that, from the algorithm's perspective, the purified evolution with the extra register $\mathcal I$ is equivalent to drawing $H$ at random once and using that same Hamiltonian throughout.
\begin{lemma}[The input register represents a fixed random Hamiltonian]\label{lem:purification}
For every $t\in[0,t_f]$,
\begin{equation}\label{eq:purification}
\Tr_{\mathcal I}\bigl(\proj{\Psi_F(t)}\bigr)=\mathbb E_{h\sim\nu_\eps^{\otimes 2^n}}\bigl[\proj{\psi_H(t)}\bigr].
\end{equation}
Also, $\ket{\Psi_{Z}(t)}=\ket{\psi_0(t)}\otimes\ket{\Omega}$, where $\ket{\psi_0(t)}$ is the algorithm state for $H=0$.
\end{lemma}
\begin{proof}
This is the input-register construction in the proof of \cite[Theorem~3.1]{brandeho2015universal}, adapted to our continuous distribution.
\end{proof}

\subsubsection{The definition of the progress function}\label{subsec:defprogfun}

In this section, we define the progress function. Regardless of whether the input is the zero Hamiltonian or our random family of Hamiltonians, the algorithm starts in the same joint state $\ket{\Psi_{Z}(0)} = \ket{\Psi_F(0)}$. A successful algorithm must eventually distinguish these two states, so this suggests using the quantity $1-\operatorname{Re}\braket{\Psi_{Z}(t)}{\Psi_F(t)}$ to measure progress. To obtain the desired bound on its rate of change however, we add a \emph{correction} term, and define the progress function as
\begin{equation}\label{eq:progress}
    {\cal P}(t) := 1-\operatorname{Re}\braket{\Psi_{Z}(t)}{\Psi_F(t)} + \frac{1}{4}\bra{\Psi_F(t)}(I_{\cal A} \otimes G)\ket{\Psi_F(t)}.
\end{equation}
Here, $G$ is a bounded nonnegative operator on the input register $\mathcal{I}$, satisfying $G\ket{\Omega} = 0$. 
The fact that $G$ is nonnegative and that $G\ket{\Omega} = 0$, suffices to ensure that ${\cal P}(t) \geq 0$ and ${\cal P}(0) = 0$. It is also enough to prove
that any successful algorithm must generate a lot of progress.  More precisely, in~\cref{sec:progress-end} we prove that for a successful algorithm
\begin{equation}\label{eq:mustprogress}
    {\cal P}(t_f) = \Omega(1).
\end{equation}
However, a more precise construction of $G$ in the correction term is necessary to make sure that the progress can not grow too large with every continuous query. In \cref{sec:progress-rate}, we give a precise construction of $G$ and prove that it satisfies  
\begin{equation}\label{eq:cannotprogresstoomuch}
    \frac{\dd}{\dd t}{\cal P}(t) \leq O(\eps^2\lambda(t)).
\end{equation}
Now, putting \cref{eq:mustprogress,eq:cannotprogresstoomuch} together we have that the total time evolution $T$ of a successful algorithm must satisfy
\begin{equation*}
    \Omega(1)=\mathcal P(t_f)=O(\eps^2)\int_0^{t_f}\lambda(t)dt=O(T\eps^2).
\end{equation*}
Rearranging, we obtain $T=\Omega(1/\eps^2)$, as desired.

\subsection{A successful algorithm must progress}\label{sec:progress-end}
In this section, we show that a successful algorithm that distinguishes the random ensemble from the zero Hamiltonian, must satisfy that at the final time the progress function is large.

\begin{lemma}\label{lem:successful-progress}
If an algorithm satisfies $\norm{\ket{\Psi_F(t_f)}-\ket{\Psi_{Z}(t_f)}}_{\cal AI}=\Omega(1)$, then its corresponding progress function satisfies ${\cal P}(t_f)= \Omega(1)$.
\end{lemma}
\begin{proof}
We abbreviate $\ket{\Delta\Psi}:=\ket{\Psi_F(t_f)}-\ket{\Psi_{Z}(t_f)}$. Thus, $\langle\Delta\Psi|\Delta\Psi\rangle=\Omega(1)$.  If we expand $\langle\Delta\Psi|\Delta\Psi\rangle$, we notice that
\begin{align*}
    \langle\Delta\Psi|\Delta\Psi\rangle = 2 - \braket{\Psi_{Z}(t_f)}{\Psi_F(t_f)} - \braket{\Psi_F(t_f)}{\Psi_{Z}(t_f)} = 2(1-\operatorname{Re}\braket{\Psi_{Z}(t_f)}{\Psi_F(t_f)}).
\end{align*}
Hence, since $G\succeq0$, we conclude that
\[
{\cal P}(t_f)\geq 1-\operatorname{Re}\braket{\Psi_{Z}(t_f)}{\Psi_F(t_f)} = \frac{\langle\Delta\Psi|\Delta\Psi\rangle}{2} = \Omega(1).
\]
\end{proof}

\subsection{The progress growth is bounded}\label{sec:progress-rate}
In this section, we give a formal definition of the operator $G$ announced in \cref{subsec:defprogfun} and we prove that the derivative of the progress function is bounded. The definition of $G$ requires some work. We first introduce the operators
\begin{align*}
    O&:L^2([-1/2,1/2])\to L^2(\R_+),\\
    R,D&:L^2(\R_+)\to L^2(\R_+).
\end{align*} 
The operators $O$ and $R$ describe the two contributions to the overlap with $\ket{b}$ discussed after \cref{lem:propertieshardextra}: the contribution from an initial vector and that from terms added during the evolution. Define
\begin{align}
    (Or)(\tau)&:=\bra{b}e^{-\id M\tau}|r\rangle,\nonumber\\
    (Ru)(\tau)&:=-\id\int_0^\tau
    \bra{b}e^{-\id M(\tau-\sigma)}\ket{b}u(\sigma)\dd\sigma,
    \label{eq:OR-definition}
\end{align}
where $b\in L^2([-1/2,1/2])$ and $M\in B(L^2([-1/2,1/2]))$ are those defined in \cref{sec:hardinst}.

The bounds in \cref{lem:propertieshardextra} ensure that both $O$ and $R$ are bounded. \cref{eq:hard-observation} directly gives $\norm{O}_{\op}^2\leq C_{\rm obs}$. To bound $R$, we define $\kappa(\tau):=\bra{b}e^{-\id M\tau}\ket{b}$ for $\tau\geq0$, and extend both $\kappa$ and $u$ by zero on negative arguments. Then $(Ru)(\tau)=-\id(\kappa*u)(\tau)$ for $\tau\geq0$, so Young's convolution inequality and
\cref{eq:hard-integrable} yield 
\[
    \norm{Ru}_{L^2(\R_+)} \leq \left(\int_0^\infty\abs{\bra{b}e^{-\id M\tau}\ket{b}}\dd\tau\right) \norm{u}_{L^2(\R_+)} \leq C\norm{u}_{L^2(\R_+)}.
\]
In other words, $\norm{R}_{\op}\leq C$.

Now, we can introduce the operator $D$, which we define as 
\begin{equation*}
    \eta^2:=1+\norm{R}_{\op}^2, \qquad D:=\eta^2 I_{L^2(\R_+)}-R^\dagger R.
\end{equation*}
We have defined $D$ so that it is self-adjoint, and satisfies $D\succeq I_{L^2(\R_+)}$, so $D^{-1}$ exists and $\norm{D^{-1}}_{\op}\leq1$. We use this $D$ to define the nonnegative operator $g$, that will act on a single input register in ${\cal I}$. Define
\[
    g:\mathbb C \ket \perp \oplus L^2([-1/2,1/2])\to \mathbb C \ket \perp \oplus L^2([-1/2,1/2])
\]
by
\begin{equation}\label{eq:g-definition}
    g:=0\oplus\left(O^\dagger O+O^\dagger RD^{-1}R^\dagger O\right).
\end{equation}
We use lowercase $g$ to distinguish this operator from the full correction $G$. The reason for this choice of~$g$ will become clear below: we prove an identity for $\bra{r}g\ket{r}$ that lets us bound the change in its expectation during evolution under $J_{\eps}$. Here we identify $\ket{r}\in L^2([-\frac{1}{2},\frac{1}{2}])$ with $0\ket{\perp}+\ket{r}$.

We recall that $\mathcal I_x=\mathbb C \ket \perp \oplus L^2([-1/2,1/2])$ and that $\mathcal I=\otimes_{x\in\{0,1\}^n}\mathcal I_x$. For each $x\in\{0,1\}^n$, let $g_x$ act as $g$ on $\mathcal I_x$ and as the identity on the other input registers. Define
\begin{equation}\label{eq:G-definition}
    G:=\sum_{x\in\{0,1\}^n}g_x
\end{equation}
on $\mathcal I$. By construction, every $g_x$ is nonnegative, so $G$ is also nonnegative. Moreover,  $\norm{G}_{\op}\leq2^n\norm{g}_{\op}<\infty$, and $G\ket{\Omega}=0$, as required in \cref{subsec:defprogfun}, where we recall that $\ket\Omega=\otimes_{x}\ket\perp$. Hence, it remains to show that with this choice of $G$ the derivative of the progress function is bounded. This requires a bit of work, and we will need to prove a few intermediate lemmas.

\subsubsection{Bounding the change in the correction term}

Consider a vector $q(\tau)\ket{\perp}+\ket{r(\tau)}$ evolving under $J_{\eps}$, with $q(\tau)\in\C$ and $\ket{r(\tau)}\in L^2([-\frac{1}{2},\frac{1}{2}])$, meaning that $$\id\frac{\dd}{\dd\tau}(q(\tau)\ket{\perp}+\ket{r(\tau)})=J_\eps(q(\tau)\ket{\perp}+\ket{r(\tau)}).$$
By \cref{eq:blockJ},
\begin{equation}\label{eq:revolution}
    \id\frac{\dd}{\dd\tau}q(\tau)=\eps\braket{b}{r(\tau)},
    \qquad
    \id\frac{\dd}{\dd\tau}\ket{r(\tau)}
    =M\ket{r(\tau)}+\eps q(\tau)\ket{b}.
\end{equation}
To see how $O$ and $R$ describe this evolution, temporarily replace $\eps q(\tau)$ by an arbitrary function $u\in L^2(\R_+)$. Setting $\ket{r(0)}=\ket{r}$, the second equation has solution
\begin{equation}\label{eq:forced-evolution}
    \ket{r_u(\tau)}
    =e^{-\id M\tau}\ket{r}
    -\id\int_0^\tau e^{-\id M(\tau-\sigma)}\ket{b}u(\sigma)\dd\sigma.
\end{equation}
By \cref{eq:OR-definition,eq:forced-evolution}, we have
\begin{equation}\label{eq:overlapwithb}
    \braket{b}{r_u(\tau)}=(O\ket{r})(\tau)+(Ru)(\tau).
\end{equation}
We want the expectation of $g$ to decrease when this overlap is large, apart from a contribution proportional to $\abs{u(\tau)}^2$. For the actual evolution, $u(\tau)=\eps q(\tau)$, so that contribution will have the required factor $\eps^2$. To achieve this, we use the following lemma.
\begin{lemma}\label{lem:g-variational}
    For every $\ket{r}\in L^2([-\frac{1}{2},\frac{1}{2}])$, the operator $g$ from \cref{eq:g-definition} satisfies
    \begin{equation}\label{eq:g-variational}
        \bra{r}g\ket{r} 
        =\sup_{u\in L^2(\R_+)}
        \left(
            \norm{O\ket{r}+Ru}_{L^2(\R_+)}^2
            -\eta^2\norm{u}_{L^2(\R_+)}^2
        \right).
    \end{equation}
    Furthermore, there exists $u\in L_2(\R_+)$ that attains the supremum.
\end{lemma}
\begin{proof}
    We manipulate the argument inside the supremum of \cref{eq:g-variational}. We recall that $D=\eta^2I-R^\dagger R$. Then,
    \begin{align*}
        \norm{O\ket{r}+Ru}^2_{L^2(\R_+)} -\eta^2\norm{u}^2_{L^2(\R_+)}&=\langle r|O^\dagger O|r\rangle+2\operatorname{Re}\langle Ru|Or\rangle+\langle Ru|Ru\rangle-\eta^2\langle u|u\rangle\\
        &=\langle r|O^\dagger O|r\rangle +2\operatorname{Re}\langle u|R^\dagger Or\rangle -\langle u|Du\rangle.
    \end{align*}
    Now, we add and subtract $\langle r|O^\dagger RD^{-1}R^\dagger O |r\rangle$ and we arrive at 
    \begin{align*}
        \norm{O\ket{r}+Ru}^2_{L^2(\R_+)} -\eta^2\norm{u}^2_{L^2(\R_+)}=&\langle r|O^\dagger O+O^\dagger RD^{-1}R^\dagger O|r\rangle\\
        &-\langle r|O^\dagger RD^{-1}R^\dagger O|r\rangle+2\operatorname{Re}\langle u|R^\dagger O|r\rangle -\langle u|Du\rangle.
    \end{align*}
    Now, we recall the definition of $g$ given in \cref{eq:g-definition} and set $v:=R^\dagger O\ket r,$ and get 
    \begin{align*}
        \norm{O\ket{r}+Ru}^2_{L^2(\R_+)} -\eta^2\norm{u}^2_{L^2(\R_+)}=&\langle r|g|r\rangle-\langle v|D^{-1}|v\rangle+2\operatorname{Re}\langle u|v\rangle -\langle u|Du\rangle.
    \end{align*}
    By identifying a square it follows that 
    \begin{align*}
        \norm{O\ket{r}+Ru}^2_{L^2(\R_+)} -\eta^2\norm{u}^2_{L^2(\R_+)}=&\langle r|g|r\rangle-\langle u-D^{-1}v|D|u-D^{-1}v\rangle.
    \end{align*}
    Since $D\succeq 0$, it follows that $\langle r|g|r\rangle$ is greater or equal than the supremum in \cref{eq:g-variational}. The equality follows by taking $u=D^{-1}v\in L^2(\R_+)$.
\end{proof}

To use this result for the algorithm, we need a bound on the derivative of the correction. If $\ket{\psi(\tau)}=e^{-\id J_{\eps}\tau}\ket{\psi}$, then
\[
    \frac{\dd}{\dd\tau}\bra{\psi(\tau)}g\ket{\psi(\tau)}
    =\bra{\psi(\tau)}\id[J_{\eps},g]\ket{\psi(\tau)}.
\]
Thus the commutator is simply the operator expressing this derivative. The following lemma bounds it in a form that we can apply to each input register.

\begin{lemma}\label{lem:single-correction}
    The operator $g$ from \cref{eq:g-definition} satisfies
    \begin{equation}\label{eq:single-commutator}
        \id[J_{\eps},g]\preceq\eta^2\eps^2\proj{\perp}-\proj{b}.
    \end{equation}
\end{lemma}
\begin{proof}
    Let $\ket{\psi}=q\ket{\perp}+\ket r$ with $\ket r\in L^2([-1/2,1/2])$. We define $\ket{\psi(\tau)}:=e^{-\id J_\eps\tau}\ket{\psi}$. Then, we have that 
    \begin{equation*}
        \langle\psi|\id[J_{\eps},g]|\psi\rangle=\left.\frac{\dd}{\dd \tau}\langle \psi(\tau)|g|\psi(\tau)\rangle\right|_{\tau=0}.
    \end{equation*}
    We define $q(\tau)$ and $\ket{r(\tau)}$ as those satisfying $\ket{\psi(\tau)}=q(\tau)\ket{\perp}+\ket{r(\tau)}$. As $g\ket{\perp}=0$, it follows that  \begin{equation*}
        \langle\psi|\id[J_{\eps},g]|\psi\rangle=\left.\frac{\dd}{\dd \tau}\langle r(\tau)|g|r(\tau)\rangle\right|_{\tau=0}=\lim_{\tau\to 0}\frac{\langle r(\tau)|g|r(\tau)\rangle-\langle r|g|r\rangle}{\tau}.
    \end{equation*}
    In order to give an upper bound on $\langle r(\tau)|g|r(\tau)\rangle-\langle r|g|r\rangle$,
    we fix $\tau>0$. By \cref{lem:g-variational}, there is $u\in L^2(\R_+)$ such that 
    \begin{equation*}
        \langle r(\tau)|g|r(\tau)\rangle\leq \norm{O\ket{r(\tau)}+Ru}_{L^2(\R_+)}^2
            -\eta^2\norm{u}_{L^2(\R_+)}^2.
    \end{equation*}
    Now, we define $\ket{r_u}$ as 
    \begin{equation*}
        \ket{r_u(\sigma)}
    :=e^{-\id M\sigma}\ket{r(\tau)}
    -\id\int_0^\sigma e^{-\id M(\sigma-\mu)}\ket{b}u(\mu)\dd\mu.
    \end{equation*}
    Then, by \cref{eq:overlapwithb} it follows that 
    \begin{equation}\label{eq:upperboundrtau}
         \langle r(\tau)|g|r(\tau)\rangle\leq \int_0^\infty \abs{\langle b|r_u(\sigma)\rangle}^2
            -\eta^2\abs{u(\sigma)}^2 \dd \sigma.
    \end{equation}
    Now, we give a complementary lower bound on $\langle r|g|r\rangle$ by defining an appropriate $u'$ and evaluating the supremum of \cref{eq:g-variational}. We define $u'\in L^2(\R_+)$ as 
    \begin{equation*}
        u'(\sigma)=\left\{\begin{array}{ll}
            \eps q(\sigma), & \sigma\leq \tau, \\
            u(\sigma-\tau) & \sigma>\tau.
        \end{array}\right.
    \end{equation*}
    We also define $\ket{r_{u'}}$ as 
    \begin{equation*}
        \ket{r_{u'}(\sigma)}
    :=e^{-\id M\sigma}\ket{r}
    -\id\int_0^\sigma e^{-\id M(\sigma-\mu)}\ket{b}u'(\mu)\dd\mu.
    \end{equation*}
    Again, by \cref{lem:g-variational,eq:overlapwithb}, we have that 
    \begin{equation}\label{eq:lowerboundr}
        \langle r|g|r\rangle\geq \int_0^\infty \abs{\langle b|r_{u'}(\sigma)\rangle}^2
            -\eta^2\abs{u'(\sigma)}^2\dd \sigma.
    \end{equation}
    Now, we manipulate $\ket{r_{u'}(\sigma)}$ for $\sigma\geq \tau$ to arrive at
    \begin{align}\nonumber
        \ket{r_{u'}(\sigma)}&=e^{-\id M\sigma}\ket{r}
    -\id\int_\tau^\sigma e^{-\id M(\sigma-\mu)}\ket{b}u(\mu-\tau)\dd\mu-\id\int_0^\tau e^{-\id M(\sigma-\mu)}\ket{b}\eps q(\mu)\dd\mu\\ \nonumber
    &=e^{-\id M\sigma}\ket{r}-\id \int_0^{\sigma-\tau} e^{-\id M((\sigma-\tau)-\mu)}\ket{b}u(\mu)\dd\mu-\id\int_0^\tau e^{-\id M(\sigma-\mu)}\ket{b}\eps q(\mu)\dd\mu\\ \nonumber
    &=e^{-iM(\sigma-\tau)}\left(e^{-iM\tau}\ket{r}-\id\int_0^\tau e^{-\id M(\tau-\mu)}\ket{b}\eps q(\mu)\dd\mu\right)-\id \int_0^{\sigma-\tau} e^{-\id M((\sigma-\tau)-\mu)}\ket{b}u(\mu)\dd\mu\\ \nonumber
    &=e^{-iM(\sigma-\tau)}\ket{r(\tau)}-\id \int_0^{\sigma-\tau} e^{-\id M((\sigma-\tau)-\mu)}\ket{b}u(\mu)\dd\mu\\
    &=\ket{r_u(\sigma-\tau)},\label{eq:ruvsru'}
    \end{align}
    where in the second line we have performed a change of variables, in the fourth line we have identified $\ket{r(\tau)}$ via \cref{eq:forced-evolution}, and in the last line we have used the definition of $\ket{r_u}$. 

    Now, by \cref{eq:upperboundrtau,eq:lowerboundr,eq:ruvsru'}, we have that 
    \begin{align*}
        \langle r(\tau)|g|r(\tau)\rangle-\langle r|g|r\rangle&\leq \int_0^\infty \abs{\langle b|r_u(\sigma)\rangle}^2
            -\eta^2\abs{u(\sigma)}^2\dd\sigma-\int_0^\infty \abs{\langle b|r_{u'}(\sigma)\rangle}^2
            -\eta^2\abs{u'(\sigma)}^2\dd\sigma\\
            &=\int_0^\infty \abs{\langle b|r_u(\sigma)\rangle}^2
            -\eta^2\abs{u(\sigma)}^2\dd\sigma\\
            &\quad -\int_\tau^\infty \abs{\langle b|r_{u}(\sigma-\tau)\rangle}^2
            -\eta^2\abs{u(\sigma-\tau)}^2\dd\sigma-\int_0^\tau \abs{\langle b|r_{u'}(\sigma)\rangle}^2
            -\eta^2\abs{u'(\sigma)}^2\dd\sigma\\
            &=\int_0^\tau \eta^2\eps^2\abs{q(\sigma)}^2-\langle r_{u'}(\sigma) |b\rangle \langle b|r_{u'}(\sigma)\rangle
            \dd\sigma.
    \end{align*}
    Now, we use that $\ket{r(\sigma)}=\ket{r_{u'}(\sigma)}$ for $\sigma\leq \tau,$ and arrive at 
    \begin{align*}
        \frac{\langle r(\tau)|g|r(\tau)\rangle-\langle r|g|r\rangle}{\tau}\leq \frac{\int_0^\tau\eta^2\eps^2\abs{q(\sigma)}^2- \langle r(\sigma) |b\rangle \langle b|r(\sigma)\rangle
            \dd\sigma}{\tau}.
    \end{align*}
    Hence, taking the limit $\tau\to 0,$ we conclude that 
    \begin{equation*}
        \left.\frac{\dd}{\dd \tau}\langle r(\tau)|g|r(\tau)\rangle \right|_{\tau=0}\leq 
            \eta^2\eps^2\abs{q}^2-\langle r|b\rangle \langle b|r\rangle=\bra{\psi}(\eta^2\eps^2\ket{\perp}\bra{\perp}-|b\rangle \langle b|)\ket{\psi},
    \end{equation*}
    as desired.
\end{proof}

\subsubsection{Bounding the progress of the algorithm}

We now have all the tools to bound the progress ${\cal P}(t)$.

\begin{lemma}\label{lem:progress-rate}
    For $G$ from \cref{eq:G-definition}, the progress function in \cref{eq:progress} is absolutely continuous and satisfies
    \[
        \frac{\dd}{\dd t}\mathcal P(t)
        \leq C_{\rm prog}\eps^2\lambda(t)
    \]
    for almost every $t\in[0,t_f]$, and we have $C_{\rm prog}:=1+\frac{\eta^2}{4}$.
\end{lemma}
\begin{proof}
    We recall that the progress function is defined as 
    \begin{equation*}
        {\cal P}(t) = 1-\operatorname{Re}\braket{\Psi_{Z}(t)}{\Psi_F(t)} + \frac{1}{4}\bra{\Psi_F(t)}(I_{\cal A} \otimes G)\ket{\Psi_F(t)}.
    \end{equation*}
    We first upper bound the derivative of the \emph{correction} term determined by $G$, and later the derivative of the \emph{infidelity} term $1-\operatorname{Re}\braket{\Psi_{Z}(t)}{\Psi_F(t)}$. By \cref{lem:single-correction}, the definition of $\mathbf J_{\eps}(t)$, and $\sum_xP_x(t)\preceq I_{\mathcal A}$,
    \begin{align}
        \id[\mathbf J_{\eps}(t),I_{\mathcal A}\otimes G]
        &=\sum_xP_x(t)\otimes\id[J_{\eps}^{(x)},g_x]\nonumber\\
        &\preceq
        \eta^2\eps^2 I_{\mathcal A\otimes\mathcal I}
        -\sum_xP_x(t)\otimes\proj{b}_x,
        \label{eq:G-commutator}
    \end{align}
    where $\proj{b}_x$ acts as $\proj{b}$ on $\mathcal I_x$ and as the identity on the other input registers. Write
    \[
        \beta(t):=
        \bra{\Psi_F(t)}
        \left(\sum_xP_x(t)\otimes\proj{b}_x\right)
        \ket{\Psi_F(t)}\geq0.
    \]
    By the definition of $\ket{\Psi_F(t)}$ in \cref{eq:purifiedevol}, we have that 
    \begin{align*}
        \frac{\dd}{\dd t}
        \bra{\Psi_F(t)}I_{\mathcal A}\otimes G\ket{\Psi_F(t)}&=\bra{\Psi_F(t)}i[K(t)\otimes I_{\mathcal I}+\lambda(t)\mathbf J_\eps(t),I_{\mathcal A}\otimes G]\ket{\Psi_F(t)}\\
        &=\lambda(t)\bra{\Psi_F(t)}i[\mathbf J_\eps(t),I_{\mathcal A}\otimes G]\ket{\Psi_F(t)}.
    \end{align*}
    Thus, by \cref{eq:G-commutator} and the fact that $\lambda(t)\geq 0,$ we have that 
    \begin{equation}\label{eq:correction-derivative}
        \frac{\dd}{\dd t}
        \bra{\Psi_F(t)}I_{\mathcal A}\otimes G\ket{\Psi_F(t)}
        \leq \lambda(t)\left(\eta^2\eps^2-\beta(t)\right).
    \end{equation}
    To control the derivative of the infidelity, we note that by \cref{eq:purifiedevol}, we have that 
    \[
        \frac{\dd}{\dd t}\braket{\Psi_{Z}(t)}{\Psi_F(t)}
        =-\id\lambda(t)
        \bra{\Psi_{Z}(t)}\mathbf J_{\eps}(t)\ket{\Psi_F(t)}.
    \]
    By \cref{lem:purification}, $\ket{\Psi_{Z}(t)}=\ket{\psi_0(t)}\otimes\ket{\Omega}$. Define
    $\ket{b_x}:=\ket{b}_x\otimes\bigotimes_{y\neq x}\ket{\perp}_y\in\mathcal I$,
    so $J_{\eps}^{(x)}\ket{\Omega}=\eps\ket{b_x}$. Triangle inequality and Cauchy--Schwarz then give
    \begin{align*}
        \abs{\bra{\Psi_{Z}(t)}\mathbf J_{\eps}(t)\ket{\Psi_F(t)}}&\leq \sum_x \abs{(\langle \psi_0(t)|\otimes \bra{\Omega})(P_x(t)\otimes J_\eps^{(x)})(P_x(t)\otimes I)|\Psi_F(t)\rangle }\\
        &=\eps \sum_x \abs{(\langle \psi_0(t)| P_x(t))  (P_x(t)\otimes \bra{b_x})|\Psi_F(t)\rangle }\\
        &\leq\eps\sum_x
        \norm{P_x(t)\ket{\psi_0(t)}}_{\mathcal A}
        \norm{(P_x(t)\otimes\bra{b_x})\ket{\Psi_F(t)}}_{\mathcal A}.
    \end{align*}
    Now, we use that $\proj{b_x}\preceq\proj{b}_x$, the orthogonality of the $P_x(t)$ to arrive and the fact that $\lambda(t)\geq 0,$ to arrive at 
    \begin{equation}\label{eq:upperboundoverlap}
                 \abs{\frac{\dd}{\dd t}\braket{\Psi_{Z}(t)}{\Psi_F(t)}}\leq\eps\lambda(t)
        \left(\sum_x\norm{P_x(t)\ket{\psi_0(t)}}_{\mathcal A}^2\right)^{1/2}
        \sqrt{\beta(t)}
        \leq\eps\lambda(t)\sqrt{\beta(t)}.
    \end{equation}
    Combining this bound with \cref{eq:correction-derivative} and the definition in \cref{eq:progress}, we obtain
    \begin{align*}
        \frac{\dd}{\dd t}\mathcal P(t)
        &\leq\lambda(t)\left(
        \eps\sqrt{\beta(t)}
        +\frac{\eta^2\eps^2}{4}
        -\frac{\beta(t)}{4}\right)\\
        &=\lambda(t)\left(
        \left(1+\frac{\eta^2}{4}\right)\eps^2
        -\left(\frac{\sqrt{\beta(t)}}{2}-\eps\right)^2\right)
        \leq C_{\rm prog}\eps^2\lambda(t),
    \end{align*}
    which completes the proof.
\end{proof}

\subsection{Putting everything together}
The lower bound now follows by comparing the final progress required by \cref{lem:successful-progress} with the rate bound in \cref{lem:progress-rate}.

\begin{theorem}\label{thm:mainresult}
    Let $0<\eps<1/2$. If an algorithm satisfies $\norm{\ket{\Psi_F(t_f)}-\ket{\Psi_{Z}(t_f)}}_{\cal AI}=\Omega(1)$, then the algorithm's total evolution time $T$ satisfies
    \[
        T=\Omega\left(\frac{1}{\eps^2}\right).
    \]
\end{theorem}
\begin{proof}
    By \cref{lem:successful-progress}, we have that such an algorithm satisfies
    \begin{equation*}
        \mathcal P(t_f)-\mathcal P(0)\geq\Omega(1).
    \end{equation*}
    On the other hand, by \cref{lem:progress-rate}, we have that 
    \begin{equation*}
        \mathcal P(t_f)-\mathcal P(0)= \int_0^{t_f} \dot{\mathcal P}(s)\dd s\leq C_{\text{prog}}\eps^2\int_0^{t_f} \lambda(s)\dd s= C_{\text{prog}}\eps^2T.
    \end{equation*}
    Putting both things together and rearranging yields the claimed result.
\end{proof}

\begin{remark}[Where forward evolution is used]
    The restriction $\lambda(t)\geq0$ is used to obtain \cref{eq:correction-derivative,eq:upperboundoverlap}. For negative $\lambda(t)$, the operator inequalities that lead to \cref{eq:correction-derivative} would reverse direction, and thus would no longer imply the upper bound on the derivative of the progress function ${\cal P}(t)$. 
\end{remark}

\section{Applications}\label{sec:corollaries}
In this section, we show how our main technical result, \cref{thm:mainresult}, implies hardness of Hamiltonian certification, locality testing and amplitude estimation.

\subsection{Locality testing does not achieve the Heisenberg limit}
Our first application is that Hamiltonian locality testing does not achieve the Heisenberg limit. For our reduction to \cref{thm:mainresult}, we need to show that the random diagonal Hamiltonian is far from local with high probability.

\begin{lemma}[The random Hamiltonian is far from local]\label{cor:randomHisFarFromLocal}
    Let $0<\alpha<1/2$, $n\in \N$, and $k\in \N_0$. Assume that $k\leq n/8$ and $\alpha^2\geq 24/2^n$. Let $H=\sum_{x\in \{0,1\}^n}h_x\ket{x}\bra{x}$ be a random Hamiltonian, where the $h_x$ are drawn independently from $\nu_\alpha$ as in \cref{def:harddist}. Then, 
    $$\Pr\left[\mathrm{dist}(H,k\textrm{-}\mathrm{local}) > \dfrac{\alpha}{2}\right]\geq \dfrac{2}{3}.$$ 
\end{lemma} 
\begin{proof}
    As $H$ is diagonal in the computational basis, 
    \begin{equation*}
        (\mathrm{dist}(H,k\textrm{-local}))^2=\sum_{S\subseteq [n]:\, |S|>k}|\widehat h(S)|^2,
    \end{equation*}
    where $\widehat h(S)=\mathbb E_x[h_x(-1)^{\sum_{i\in S}x_i}]$ are the Fourier coefficients of $h:\{0,1\}^n\to \R$. We now note that due to Parseval's identity
    \begin{equation*}
        (\mathrm{dist}(H,k\textrm{-local}))^2=\sum_{S\subseteq [n]}|\widehat h(S)|^2-\sum_{S\subseteq [n]:\, |S|\leq k}|\widehat h(S)|^2=\underbrace{\sum_{x\in\{0,1\}^n}\frac{|h_x|^2}{2^n}}_{(*)}-\underbrace{\sum_{S\subseteq [n]:\, |S|\leq k}|\widehat h(S)|^2}_{(**)}.
    \end{equation*}
    We continue by showing that $(*)$ is large and $(**)$ is small with high probability, beginning with the former. Indeed, by linearity of expectation and \cref{lem:propertiesharddist}, we have that $\mathbb E_H[(*)]=\alpha^2$. Now, by independence of the $h_x$ and the fact that $h_x\in [-1,1]$, we have that $\var_H[(*)]=\sum_{x}\var_H[h_x^2]/2^{2n}\leq \sum_{x}\mathbb E_H[h_x^2]/2^{2n}\leq \alpha^2/2^n$. Now, by Chebyshev's inequality and the assumption $\alpha^2\geq 24/ 2^n$, we have that 
    \begin{equation}\label{eq:concentration1}
        \Pr\left[(*)\leq \dfrac{\alpha^2}{2}\right]\leq \frac{\alpha^2/2^n}{\alpha^4/4}=\frac{4}{2^n\alpha^2}\leq \frac{1}{6}.    
    \end{equation}
    Next, we show that $(**)$ is small with high probability. Indeed, in expectation $(**)$ behaves as
    \begin{align*}
        \mathbb E_H[(**)]&=\sum_{S\subseteq [n]:\, |S|\leq k}\mathbb E_{H}[|\widehat h(S)|^2]=\sum_{S\subseteq [n]:\, |S|\leq k}\mathbb E_{H}[\mathbb E_{x,y}[h_xh_y(-1)^{\sum_{i\in S}x_i+y_i}]]\\
        &=\sum_{S\subseteq [n]:\, |S|\leq k}\alpha^2\mathbb E_{x,y}[(-1)^{\sum_{i\in S}x_i+y_i}\delta_{x,y}]=\frac{\alpha^2}{2^n}\sum_{S\subseteq [n]:\, |S|\leq k} 1\\
        &\leq \dfrac{\alpha^2}{24},
    \end{align*}
    where in the second line we have used that $\mathbb E[h_x^2]=\alpha^2$, $\mathbb E[h_x]=0$ (see \cref{lem:propertiesharddist}), and that $h_x$ and $h_y$ are independent for $x\neq y;$ and in the third line that $k\leq n/8$. Now, by Markov's inequality we have that 
    \begin{equation}\label{eq:concetration2}
        \Pr\left[(**)\geq \dfrac{\alpha^2}{4}\right]\leq \dfrac{1}{6}.
    \end{equation}
    Combining \cref{eq:concentration1,eq:concetration2} one has that 
    \begin{align*}
        \Pr\left[\mathrm{dist}(H,k\textrm{-local}) >  \dfrac{\alpha}{2}\right]&=\Pr\left[(\mathrm{dist}(H,k\textrm{-local}))^2 > \dfrac{\alpha^2}{4}\right]\\&\geq 1-\Pr\left[(*)\leq \dfrac{\alpha^2}{2}\right]-\Pr\left[(**)\geq \dfrac{\alpha^2}{4}\right]\geq \dfrac{2}{3},
    \end{align*}
    as desired.
\end{proof}

Next, we  show that if an algorithm is capable of performing locality testing, then the corresponding purified final states are far away from each other at the final time.

\begin{corollary}\label{cor:hardnesslocality}
    Let $0 < \eps < 1/4,$ $n\in \N$ and $k\in \N$. Assume that $k\leq n/8$ and $n \geq 7$. Testing $k$-locality with distance threshold $\eps$ and success probability at least $9/10$ requires total forward evolution time at least 
    \[
        T=\Omega\left(\min\left\{\frac{1}{\eps^2},\frac{2^{n/2}}{\eps}\right\} \right).
    \]
\end{corollary}
\begin{proof}
    Set $\alpha=\max\{2\eps,\sqrt{24/2^n}\}.$ Consider the random Hamiltonian $H$ ensemble of~\cref{sec:hardinst} for $\nu_\alpha$. Then, consider the random Hamiltonian $\widetilde H=(2\eps/\alpha)H$. By \cref{cor:randomHisFarFromLocal},$\widetilde H$ is $\eps$-far from local with probability at least 2/3. Also note that $2\eps/\alpha\leq 1$, so $\norm{\widetilde H}_{\op}\leq \norm{H}_{\op}\leq 1$. 
    
    Consider an algorithm that decides whether a Hamiltonian is $k$-local or $\eps$-far from it. When considering this algorithm on input $\widetilde H$, we can instead assume that the input is $H$, and then rescale the total evolution time by $(2\eps/\alpha).$ Let $T$ be the time evolution of the algorithm when evolving according to $H$, and let $\widetilde T$ the time evolution when evolving according to $\widetilde H$. Consequently, $\widetilde T=(\alpha/2\eps)T$.
    
    Let $\ket{\Psi_F(t)}$ and $\ket{\Psi_Z(t)}$ the purified time evolutions evolving according to $H$ and the zero Hamiltonian, respectively.  Let the final binary measurement of the algorithm be determined by $\{P_Z,P_F\}$. Since the zero Hamiltonian is $k$-local,  \cref{lem:purification} implies $$\bra{\Psi_Z(t_f)}P_Z\otimes I_{\mathcal I}\ket{\Psi_Z(t_f)}\geq 9/10.$$
    
    As $\dist(\widetilde H,k\text{-local})>\eps$ with probability at least $2/3$, and since for every such Hamiltonian, our algorithm succeeds with probability at least $9/10$, we have by \cref{lem:purification} that
    \[
    \bra{\Psi_F(t_f)}(P_F\otimes I_{\mathcal I})\ket{\Psi_F(t_f)} =\mathbb E_H\left[\bra{\psi_H(t_f)}P_F\ket{\psi_H(t_f)}\right] \geq \frac{2}{3}\cdot\frac{9}{10} =\frac{3}{5}.
    \]
    Then,
    \[
    \bra{\Psi_F(t_f)}(P_F\otimes I_{\mathcal I})\ket{\Psi_F(t_f)} -\bra{\Psi_{Z}(t_f)}(P_F\otimes I_{\mathcal I})\ket{\Psi_{Z}(t_f)} \geq \frac{3}{5}-\frac{1}{10}\geq \frac{1}{2}.
    \]
    
    We abbreviate $\ket{\Delta\Psi}:=\ket{\Psi_F(t_f)}-\ket{\Psi_{Z}(t_f)}$. Rearranging the difference above and applying Cauchy--Schwarz gives
    \begin{align*}
    \frac{1}{2} &\leq\abs{\bra{\Delta\Psi} (P_F\otimes I_{\mathcal I})\ket{\Psi_F(t_f)}+\bra{\Psi_{Z}(t_f)}(P_F\otimes I_{\mathcal I})\ket{\Delta\Psi}} \leq2\sqrt{\langle\Delta\Psi|\Delta\Psi\rangle}.
    \end{align*}
    Hence, $\norm{\Delta\Psi}_{\cal AI}\geq 1/4$, allowing us to apply \cref{thm:mainresult} yielding 
    \[
    T = \Omega(1/\alpha^2).
    \]
    As $\widetilde T=(\alpha/2\eps)T,$ it follows that 
    \begin{equation*}
        \widetilde T=\Omega\left(\frac{1}{\alpha\eps}\right)=\Omega\left(\min\left\{\frac{1}{\eps^2},\frac{2^{n/2}}{\eps}\right\}\right).
    \end{equation*}
\end{proof}

\subsection{Hamiltonian certification does not achieve the Heisenberg limit}
For our second application, we show that Hamiltonian certification does not achieve the Heisenberg limit. For our reduction to \cref{thm:mainresult}, we need to show that the random diagonal Hamiltonian is far from the zero Hamiltonian with high probability. To avoid that this ``farness'' is a result from a term proportional to the identity (which cannot be detected by the Hamiltonian time evolution), we impose the random Hamiltonian to be traceless. This does not affect the performance of any algorithm.

\begin{lemma}[The random traceless Hamiltonian is far from zero]\label{cor:randomHisFarFromZero}
    Let $0 < \alpha < 1/2$ and $n\in \N$. Assume that $\alpha^2\geq 24/2^n$. For $H=\sum_{x\in \{0,1\}^n}h_x\ket{x}\bra{x}$, where the $h_x$ are drawn independently from $\nu_\alpha$ as in \cref{def:harddist},
    \[
    \overline{h}=\frac{1}{2^n}\sum_{x\in\{0,1\}^n} h_x,\qquad \widetilde{H}=\frac{1}{2}(H-\overline{h} I).
    \]
    Then $\widetilde{H}$ is traceless, $\norm{\widetilde{H}}_{\op}\leq1$, and
    \[
    \Pr\left[\norm{\widetilde{H}}_2>\frac{\alpha}{4}\right]\geq\frac{2}{3}.
    \]
    \end{lemma}
    \begin{proof}

    By the definition of $\overline{h}$, 
    \[
    \Tr[\widetilde{H}]=\frac{1}{2}\left(\sum_{x\in\{0,1\}^n} h_x-2^n\overline{h}\right)=0,
    \]
    so $\widetilde{H}$ is traceless. Moreover, since both $\abs{h_x},\abs{\overline h}\leq 1$, the operator-norm bound is immediate. 
    
    Expanding the normalized Frobenius norm gives
    \begin{equation}\label{eq:expand-frobnom}
    \begin{split}
    \norm{\widetilde{H}}_2^2 &= \frac{1}{4 \cdot 2^n}\sum_{x\in\{0,1\}^n}(h_x - \overline{h})^2.
    \end{split}
    \end{equation}
    A $0$-local Hamiltonian is a scalar multiple of the identity, so $\dist(H,0\text{-local})^2$ is 
    \[
    \dist(H,0\text{-local})^2 = \min_{c} \norm{H-cI}_2^2 = \min_{c} \frac{1}{2^n}\sum_x(h_x-c)^2 =\frac{1}{2^n}\sum_{x\in\{0,1\}^n}(h_x-\overline h)^2+  \min_{c} \left\{(c-\overline h)^2\right\}.
    \]
    This minimum is attained at $c = \overline{h}$, meaning by \cref{eq:expand-frobnom} that
    \begin{equation}\label{eq:distIsFrobnorm}
    \begin{split}
    \dist(H,0\text{-local}) = 2\norm{\widetilde{H}}_2.
    \end{split}
    \end{equation}
    Hence, the lemma follows from \cref{cor:randomHisFarFromLocal} by taking $k=0$.
\end{proof}

\begin{corollary}\label{cor:hardnesscertification}
    Let $0<\eps<1/8$ and $n\in\N$. Assume that $n \geq 7$. Hamiltonian certification with distance threshold $\eps$ and success probability at least $9/10$ requires total forward evolution time at least
    \[
    T = \Omega\left(\min\left\{\frac{1}{\eps^2},\frac{2^{n/2}}{\eps}\right\} \right).
    \]
    \end{corollary}
\begin{proof}
    Set $\alpha=\max\{4\eps,\sqrt{24/2^n}\}.$ Consider the random Hamiltonian $H$ ensemble of~\cref{sec:hardinst} for~$\nu_\alpha$. Then, consider the random Hamiltonian $\widetilde H=(4\eps/\alpha)H$. By \cref{cor:randomHisFarFromZero}, the traceless version of $\widetilde H$ is $\eps$-far from the zero Hamiltonian with probability at least 2/3. Also note that $4\eps/\alpha\leq 1$, so $\norm{\widetilde H}_{\op}\leq \norm{H}_{\op}\leq 1$. 
    
    Consider an algorithm that decides whether a Hamiltonian is the zero Hamiltonian or $\eps$-far from it. When considering this algorithm on input $\widetilde H$, we can instead assume that the input is $H$, and then rescale the total evolution time by $(4\eps/\alpha).$ Let $T$ be the time evolution of the algorithm when evolving according to $H$, and let $\widetilde T$ the time evolution when evolving according to $\widetilde H$. Consequently, $\widetilde T=(\alpha/4\eps)T$.
    
    Let the final binary measurement of the protocol be determined by $\{P_Z,P_F\}$, and let $\ket{\Psi_Z(t_f)}$ and $\ket{\Psi_F(t_f)}$ be its purified final states when evolving according to the zero Hamiltonian and the random ensemble, respectively. As in the proof of \cref{cor:hardnesslocality}, we have
    \begin{align*}
        \bra{\Psi_Z(t_f)}(P_Z\otimes I_{\mathcal I})\ket{\Psi_Z(t_f)} &\geq\frac{9}{10}, \\
        \bra{\Psi_F(t_f)}(P_F\otimes I_{\mathcal I})\ket{\Psi_F(t_f)} &\geq  \frac{3}{5},
    \end{align*}
    and 
    \[
    \bra{\Psi_F(t_f)}(P_F\otimes I_{\mathcal I})\ket{\Psi_F(t_f)} -\bra{\Psi_Z(t_f)}(P_F\otimes I_{\mathcal I})\ket{\Psi_Z(t_f)} \geq\frac{3}{5}-\frac{1}{10}=\frac{1}{2}.
    \]
    Applying \cref{thm:mainresult} therefore yields 
    \[
    T = \Omega(1/\alpha^2).
    \]
    Finally, as $\widetilde T=(\alpha/4\eps)T$, it follows that 
    \[
    \widetilde T=\Omega\left(\frac{1}{\alpha\eps}\right)=\Omega\left(\min\left\{\frac{1}{\eps^2},\frac{2^{n/2}}{\eps}\right\}\right).
    \]
\end{proof}

\subsection{Amplitude estimation requires inverses}

For our final application, we show that the properties of \cref{lem:propertiesharddist} imply that the family of random Hamiltonians can be distinguished from the zero Hamiltonian via amplitude estimation. For this, we define
\begin{equation*}
    \ket{\psi_{\text{ae}}}=\frac{1}{2^{n/2}}\sum_{x\in\{0,1\}^n}\ket{x}.
\end{equation*}
We apply $U_H(1)=e^{-\id H}$ to this state and estimate the amplitude of $U_H(1)\ket{\psi_{\mathrm{ae}}}$ in the subspace orthogonal to $\ket{\psi_{\mathrm{ae}}}$. For $H=0$, we have $U_H(1)=I$, so this amplitude is zero. For our random family, we will show that it is at least $\Omega(\eps)$ with constant probability. To express (the absolute value of) this amplitude, we define
\[
\Pi_{\mathrm{ae}}:=I-\proj{\psi_{\mathrm{ae}}}, \qquad a(H):=\sqrt{\bra{\psi_{\mathrm{ae}}}U_H(1)^\dagger \Pi_{\mathrm{ae}}U_H(1)\ket{\psi_{\mathrm{ae}}}}.
\]
For $a(H)>0$, the normalized component orthogonal to $\ket{\psi_{\mathrm{ae}}}$ is
\[
\ket{\psi_H}:=\frac{\Pi_{\mathrm{ae}}U_H(1)\ket{\psi_{\mathrm{ae}}}}{a(H)}.
\]
Thus,
\[
U_H(1)\ket{\psi_{\mathrm{ae}}} =a(H)\ket{\psi_H} +\bra{\psi_{\mathrm{ae}}}U_H(1)\ket{\psi_{\mathrm{ae}}} \ket{\psi_{\mathrm{ae}}}.
\]

\begin{corollary}\label{cor:randomHFarFromZeroViaAE}
    Let $0 < \eps < 1/16$ and $n\in \N$. Assume that $n \geq 7$. Given a known state $\ket{\psi}$ and projector $\Pi$, and continuous-time forward access to an unknown Hamiltonian $H$ with $\norm{H}_{\op}\leq1$, estimating
    $\sqrt{\bra{\psi}U_H(1)^\dagger\Pi U_H(1)\ket{\psi}}$ up to additive error $\eps$ with success probability at least $9/10$ requires, in the worst case, total evolution time
    \[
    T=\Omega\left(\min\left\{\frac{1}{\eps^2},\frac{2^{n/2}}{\eps}\right\} \right).
    \]
\end{corollary}
\begin{proof}
    Set $\alpha=\max\{6\sqrt{2}\eps,24/2^{n/2}\}.$    We take $\ket{\psi} = \ket{\psi_{\mathrm{ae}}},\, \Pi = \Pi_{\mathrm{ae}}$. Let $H=\sum_{x\in \{0,1\}^n}h_x\ket{x}\bra{x}$, where the $h_x$ are drawn independently from $\nu_\alpha$ as in \cref{def:harddist}. This implies that $U=U_H(1)=\sum_{x}e^{-ih_x}\ket x\bra x$ is also diagonal and it follows that 
    \begin{align*}
        |\langle \psi_{\mathrm{ae}}|U|\psi_{\mathrm{ae}}\rangle |^2&=\frac{1}{4^n}\abs{\sum_{x,y\in\{0,1\}^n}\bra y U_H(1)\ket x}^2=\frac{1}{4^n}\abs{\sum_{x\in\{0,1\}^n}e^{-ih_x}}^2=\frac{1}{4^n}\sum_{x,y\in\{0,1\}^n}e^{i(h_x-h_y)}\\
        &=\frac{1}{4^n}\sum_{x,y\in\{0,1\}^n}\mathrm{cos}((h_x-h_y)),
    \end{align*}
    where in the second line we have used that the sine is antisymmetric. Hence, we have, using $\cos(x)\leq 1-x^2/4$ for $|x|\leq 2$, that
    \begin{equation}\label{eq:expandipampest}
        \bra{\psi_{\mathrm{ae}}}U_H^\dagger (1)\Pi_{\mathrm{ae}}U_H(1)\ket{\psi_{\mathrm{ae}}}=1-|\langle \psi_{\mathrm{ae}}|U|\psi_{\mathrm{ae}}\rangle |^2\geq \frac{1}{4 \cdot 4^n}\sum_{x,y\in\{0,1\}^n}(h_x-h_y)^2.
    \end{equation}
    Define $\overline{h}:= \frac{1}{2^n} \sum_x h_x$. Then by expanding the double sum in \cref{eq:expandipampest} we obtain
    \begin{align*}
    \frac{1}{4\cdot4^n}\sum_{x,y \in\{0,1\}^n}(h_x-h_y)^2 &=\frac{1}{2\cdot4^n} \left(2^n\sum_{x \in\{0,1\}^n}h_x^2-\left(\sum_{x \in\{0,1\}^n}h_x\right)^2\right)  =\frac{1}{2}\left(\frac{1}{2^n}\sum_{x \in\{0,1\}^n}h_x^2-\overline h^2\right)\\
    &=\frac{1}{2\cdot2^n}\sum_{x \in\{0,1\}^n}(h_x-\overline h)^2\\
    &=\frac{1}{2}\dist(H,0\text{-local})^2.
    \end{align*}
    where the last equality follows from \cref{eq:distIsFrobnorm}. Hence, applying \cref{cor:randomHisFarFromLocal} with $k=0$ gives
    \[
    \Pr_H\left[\sqrt{\bra{\psi_{\mathrm{ae}}}U_H(1)^\dagger \Pi_{\mathrm{ae}}U_H(1)\ket{\psi_{\mathrm{ae}}}}>\frac{\alpha}{2\sqrt{2}}\right] \geq\Pr_H\left[\dist(H,0\text{-local})>\frac{\alpha}{2}\right] \geq\frac{2}{3}.
    \]

    Now, we define $\widetilde H=(6\sqrt{2}\eps/\alpha)H$. Note that $\norm{\widetilde H}_{\op}\leq (6\sqrt{2}\eps/\alpha)\norm{H}_{\op}\leq 1\cdot 1=1$. Consider now an algorithm that queries $U_{\widetilde H}(\alpha/6\sqrt{2})$ and performs amplitude-estimation of $$\sqrt{\bra{\psi_{\mathrm{ae}}}U_{\widetilde H}(\alpha/6\sqrt{2})^\dagger \Pi_{\mathrm{ae}}U_{\widetilde H}(\alpha/6\sqrt{2})\ket{\psi_{\mathrm{ae}}}}=\sqrt{\bra{\psi_{\mathrm{ae}}}U_H(1)^\dagger \Pi_{\mathrm{ae}}U_H(1)\ket{\psi_{\mathrm{ae}}}}$$ up to error $\eps$. When considering the algorithm on input $U_{\widetilde H}(\alpha/6\sqrt{2}\eps)$, we can instead assume that the input is $U_H(1)$, and then rescale the total evolution time by $6\sqrt{2}\eps/\alpha.$ Let $T$ be the time evolution of the algorithm when evolving according to $H$, and let $\widetilde T$ the time evolution when evolving according to $\widetilde H$. Consequently, $\widetilde T=(\alpha/6\sqrt{2}\eps)T$.

    We now turn this algorithm into an algorithm that distinguishes the zero Hamiltonian from our random family. To do that, let $\widetilde{a}$ be the algorithm's estimate of $\sqrt{\bra{\psi_{\mathrm{ae}}}U_H(1)^\dagger \Pi_{\mathrm{ae}}U_H(1)\ket{\psi_{\mathrm{ae}}}}$, meaning
    \[
    \Pr\left[\abs{\widetilde{a}-\sqrt{\bra{\psi_{\mathrm{ae}}}U_H(1)^\dagger \Pi_{\mathrm{ae}}U_H(1)\ket{\psi_{\mathrm{ae}}}}}\leq \eps\right]\geq\frac{9}{10}
    \]
    for every $H$. Then, by outputting zero if $\widetilde{a}\leq 3\eps/2$, and far otherwise, we claim that the algorithm serves for distinguishing the zero Hamiltonian from the random ensemble. For $H=0$, we have $\sqrt{\bra{\psi_{\mathrm{ae}}}U_H(1)^\dagger \Pi_{\mathrm{ae}}U_H(1)\ket{\psi_{\mathrm{ae}}}}=0$, so a successful estimate satisfies $\widetilde{a}\leq\eps<3\eps/2$. For our random family, with probability at least $\frac{2}{3}$, we have $\sqrt{\bra{\psi_{\mathrm{ae}}}U_H(1)^\dagger \Pi_{\mathrm{ae}}U_H(1)\ket{\psi_{\mathrm{ae}}}}>\alpha/(2\sqrt{2})\geq 3\eps$. For every such Hamiltonian, a successful estimate satisfies
    \[
    \widetilde{a}\geq \sqrt{\bra{\psi_{\mathrm{ae}}}U_H(1)^\dagger \Pi_{\mathrm{ae}}U_H(1)\ket{\psi_{\mathrm{ae}}}}-\eps >(3-1)\eps>\frac{3\eps}{2}.
    \]
    
    The binary algorithm outputs zero with probability at least $9/10$ on the zero Hamiltonian, and far with probability at least $9/10\cdot 2/3=3/5$ on each of these random inputs. In other words, if the purified time evolution when the algorithm inputs inputs the zero Hamiltonian and $H$ are $\ket{\Psi_Z(t)}$ and $\ket{\Psi_F(t)}$, respectively, and if the final binary measurement determined by $\{P_Z,P_F\}$, we have
    \begin{align*}
        \bra{\Psi_Z(t_f)}(P_Z\otimes I_{\mathcal I})\ket{\Psi_Z(t_f)} &\geq\frac{9}{10}, \\
        \bra{\Psi_F(t_f)}(P_F\otimes I_{\mathcal I})\ket{\Psi_F(t_f)} &\geq  \frac{3}{5},
    \end{align*}
    and 
    \[
    \bra{\Psi_F(t_f)}(P_F\otimes I_{\mathcal I})\ket{\Psi_F(t_f)} -\bra{\Psi_Z(t_f)}(P_F\otimes I_{\mathcal I})\ket{\Psi_Z(t_f)} \geq\frac{3}{5}-\frac{1}{10}=\frac{1}{2}.
    \]
    Applying \cref{thm:mainresult} therefore yields 
    \[
    T = \Omega(1/\alpha^2).
    \]
    Finally, as $\widetilde T=(\alpha/6\sqrt{2}\eps)T$, it follows that 
    \[
    \widetilde T=\Omega\left(\frac{1}{\alpha\eps}\right)=\Omega\left(\min\left\{\frac{1}{\eps^2},\frac{2^{n/2}}{\eps}\right\}\right).
    \]\qedhere
\end{proof}

\subsection*{Acknowledgements}
We thank Richard R. Allen, Matthias C. Caro, Yanlin Chen, Tom Gur, and Lowe Angus Lowe for helpful discussions and for sharing their concurrent work with us~\cite{allen2026power}. We also thank Arne Darras, Benoît Dubus and Jérémie Roland for useful discussions related to the continuous-time query model. F.E.G. is supported by the French National Research Agency award number “ANR-22- PNCQ-0002”. J.L. acknowledges support from the Harvard Quantum Initiative and the Kwanjeong Foundation, and this work benefited from interactions with the NSF AI Institute for Artificial Intelligence and Fundamental Interactions (IAIFI), which is supported by the National Science Foundation under Cooperative Agreement PHY-2525568. S.Z. was supported by the French PEPR integrated projects EPIQ (ANR-22-PETQ-0007) and HQI (ANR-22-PNCQ-0002). 
\paragraph{AI statement.} The construction of the hard distribution and the proof strategy for bounding the rate of change of the adversary progress function were initially proposed by ChatGPT~$5.5$ during discussions with the authors in June 2026. AI tools were also used to assist with the exposition. The authors checked and revised the resulting arguments and take full responsibility for the mathematical claims and the final manuscript.

\bibliographystyle{alpha}
\bibliography{references}

\end{document}